\documentclass[11pt]{article}
\usepackage{amsmath}
\usepackage{amsfonts}
\usepackage{enumerate}
\usepackage{appendix}%附录
\usepackage{titlesec}%附录 格式
\usepackage{amsthm}
\usepackage[numbers,sort&compress]{natbib}
\usepackage{epsfig}

\usepackage{epstopdf}
\usepackage{graphicx}%表格宏包
\usepackage{pdflscape}     % 横向页面支持
\usepackage{makecell}       % 垂直拆分单元格
\usepackage{float} %设置图片浮动位置的宏包
\usepackage{subfigure} % 插入多图时用子图显示的宏包
\usepackage{booktabs}% 专业表格包
\usepackage{stfloats}
\usepackage{caption}
\usepackage{multirow}
\usepackage[perpage]{footmisc}
\usepackage{mathrsfs}
\usepackage{amssymb}
\usepackage{slashed}
\usepackage{amsmath,amssymb,amsfonts}

\usepackage{textcomp} % 添加对特殊字符的支持
\usepackage{tikz} %绘制图形

\newtheoremstyle{italicstyle}% 名称
  {\topsep}   % 上方间距
  {\topsep}   % 下方间距
  {\itshape}  % 正文字体（斜体）
  {0pt}       % 缩进
  {\bfseries} % 标题字体
  {}          % 标题后的标点（空表示无句点）
  {5pt plus 1pt minus 1pt} % 标题后间距
  {}          % 标题说明

\newtheoremstyle{normalstyle}% 名称
  {\topsep}   % 上方间距
  {\topsep}   % 下方间距
  {\normalfont} % 正文字体（正体）
  {0pt}       % 缩进
  {\bfseries} % 标题字体
  {}          % 标题后的标点（空表示无句点）
  {5pt plus 1pt minus 1pt} % 标题后间距
  {}          % 标题说明

\theoremstyle{italicstyle}
\newtheorem{theorem}{Theorem}
\makeatletter
\newenvironment{theoremprime}[1][]
  {%
    \if\relax\detokenize{#1}\relax
      \edef\currenttheorem{\thetheorem}%
      \addtocounter{theorem}{-1}%
      \let\oldthetheorem\thetheorem%
      \renewcommand{\thetheorem}{\currenttheorem\textquotesingle}%
      \begin{theorem}%
    \else
      \edef\currenttheorem{\thetheorem}%
      \addtocounter{theorem}{-1}%
      \let\oldthetheorem\thetheorem%
      \renewcommand{\thetheorem}{#1\textquotesingle}%
      \begin{theorem}%
    \fi
  }
  {%
    \end{theorem}%
    \renewcommand{\thetheorem}{\oldthetheorem}%
  }
\makeatother
\newtheorem{lemma}{Lemma}

\theoremstyle{normalstyle}
\newtheorem{definition}{Definition}
\makeatletter
\newenvironment{definitionprime}[1][]
  {%
    \if\relax\detokenize{#1}\relax
      \edef\currentdefinition{\thedefinition}%
      \addtocounter{definition}{-1}%
      \let\oldthedefinition\thedefinition%
      \renewcommand{\thedefinition}{\currentdefinition\textquotesingle}%
      \begin{definition}%
    \else
      \edef\currentdefinition{\thedefinition}%
      \addtocounter{definition}{-1}%
      \let\oldthedefinition\thedefinition%
      \renewcommand{\thedefinition}{#1\textquotesingle}%
      \begin{definition}%
    \fi
  }
  {%
    \end{definition}%
    \renewcommand{\thedefinition}{\oldthedefinition}%
  }
\makeatother
\newtheorem{remark}{Remark}

\newtheorem{example}{Example}

\def\whitebox{{\hbox{\hskip 1pt
        \vrule height 6pt depth 1.5pt
        \lower 1.5pt\vbox to 7.5pt{\hrule width
                  3.2pt\vfill\hrule width 3.2pt}%
        \vrule height 6pt depth 1.5pt
        \hskip 1pt } }}

\usepackage{titling}

\pretitle{\begin{flushleft}\LARGE\bfseries}
\posttitle{\par\end{flushleft}}
\preauthor{\begin{flushleft}\large\bfseries}
\postauthor{\par\end{flushleft}}
\predate{\begin{flushleft}\normalsize}
\postdate{\par\end{flushleft}}

\date{}
\title{Geometric construction of $k$-optimal locally repairable codes\thanks{Supported by Science Research Project of Hebei Education Department under Grant JCZX2026032, Natural Science Foundation of Hebei Province under Grant A2025205023 and A2023205045.}}
\author{{Yi Fu\textsuperscript{a}, \ Xiuling Shan\textsuperscript{abc}\thanks{Corresponding author: X. Shan, xiulingshan@hebtu.edu.cn}   }  \\
\vspace{0.2cm}
\scriptsize{\textsuperscript{a}School of Mathematical Sciences, Hebei Normal University, Shijiazhuang 050024, China}\\
\textsuperscript{b}Hebei Research Center of the Basic Discipline Pure Mathematics, Shijiazhuang 050024, China \\
\textsuperscript{c}Hebei Key Laboratory of Computational Mathematics and Applications, Shijiazhuang 050024, China }

\begin{document}

\maketitle

\baselineskip .5truecm

%\begin{abstract}
\noindent {\bf Abstract } \\
A linear code is referred to as a locally repairable code (LRC) with locality $r$ if any erased code symbol can be recovered by accessing at most $r$ other code symbols. LRCs are highly desirable for distributed storage systems to enhance repair efficiency. In this paper, we investigate LRCs with disjoint repair sets via the parity-check matrix method. Firstly, we propose a novel concept of the $s$-Pasch configuration and present a geometric characterization for the existence of LRCs with minimum distance $5$ and locality $3$. Subsequently, we construct $k$-optimal LRCs by exploiting the point-line relationship in $PG(2,q)$. Finally, a family of $q$-ary $k$-optimal LRCs with minimum distance $6$ and general locality $r$ is constructed using partial $r$-spreads.

\vspace{0.2cm}


%\end{abstract}

\baselineskip 15pt

\section{Introduction}

Nowadays, large-scale cloud storage and distributed file
systems have scaled to such massive proportions that disk failures have become the norm. To safeguard data against disk failures, the most straightforward solution is replication. But this incurs substantial storage overhead. As an alternative, $[n,k]$ storage codes have thus been introduced: they encode $k$ information symbols into $n$ coded symbols, which are then distributed across $n$ disks. Among them, the maximum distance separable (MDS) codes are preferred, as they yield dramatic improvements in both redundancy and reliability. However, MDS codes require accessing $k$ surviving nodes for recovery, which becomes highly inefficient in large-scale systems with large $k$.

To balance storage efficiency and access performance, Huang et al. \cite{HCL} studied codes that minimized repair-node participation. Subsequently, Gopalan et al. \cite{GHSY} introduced codes with locality. Specifically, the $i$-th coordinate of a code is said to have \emph{locality $r$} if it can be recovered by accessing at most $r$ other coordinates. An $[n,k,d]$ linear code $\mathcal{C}$ is referred to as a \emph{locally repairable code} (LRC) \emph{with locality $r$} if each symbol of $\mathcal{C}$ has locality $r$, which is denoted as an $(n,k,d;r)$-LRC. When $r\ll k$, this type of code significantly reduces the disk I/O complexity during the repair process.

For $q$-ary $(n,k,d;r)$-LRCs, Gopalan et al. \cite{GHSY} first established the Singleton-like bound
\begin{equation}\label{d}
d\leq n-k - \left\lceil \frac{k}{r}\right\rceil+2.
\end{equation}
An LRC achieving the equality in \eqref{d} for given $n$, $k$ and $r$ is called \emph{$d$-optimal}. Notably, for $q$-ary $d$-optimal $(n,k,d;r)$-LRCs with fixed $q$, $r$ and $d \geq 5$, the code rate is strictly less than $\frac{r}{r+1}$. The construction of $d$-optimal LRCs has been extensively studied, as evidenced by \cite{GHSY,GXY,FFCX,FTFC,CFXH,TB,JMX,LMX,LXY,HXSC}.

In order to handle multiple erasures, LRCs have been extended through two main approaches. The first approach, introduced by Prakash et al. \cite{PKLK}, employs a single repair set tolerating multiple erasures. Numerous results related to this case are presented in \cite{TFWF,QFF,FF,KWG,QZF,TZ}. Wang et al. \cite{WZ} proposed the second approach that employs disjoint repair sets, a concept termed availability in \cite{RPDV}, to provide parallel read capability. Subsequent developments dedicated to this case include works such as \cite{CCFT,TZSP,HHCY,JC,CMST,ZK,ZLS}.

In \cite{CM}, Cadambe and Mazumdar derived the alphabet-size-dependent C-M bound for $q$-ary LRCs
\begin{equation}\label{cm}
k\leq \min_{t\in\mathbb Z^+}\left[tr+k_{opt}^{(q)}(n-t(r+1),d)\right],
\end{equation}
where $k_{opt}^{(q)}(n,d)$ denotes the maximum dimension of an $[n,k,d]$ code over $\mathbb F_q$. They also showed binary simplex codes can attain the C-M bound, and subsequent works \cite{SZ,LYRF} constructed binary LRCs achieving the bound \eqref{cm}. Unfortunately, it is generally difficult to determine the value of $k_{opt}^{(q)}(n,d)$, which limits the practical application of the C-M bound.

Recently, Wang et al. \cite{WZL} have derived a Hamming-like bound via the sphere-packing method for $(n,k,d;r)$-LRCs where $r+1$ divides $n$. An $(n,k,d;r)$-LRC is referred to as \emph{$k$-optimal} if $k$ achieves the maximum integer allowed by this bound. Ma et al. \cite{MG} proposed a class of $k$-optimal binary LRCs with $d=6$, $r=2^b$, which included some results given in \cite{WZL}, and two classes for $d=6$, $r=3$ by partial spreads. Using Hamming codes, Kim et al. \cite{KN1,KN2} constructed two classes of binary $k$-optimal LRCs with $d\geq6$ and $r =2,3$. By employing $q$-Steiner systems, sunflowers in projective geometry and difference sets in finite fields, Fang et al. \cite{FCX} developed two constructions of perfect LRCs with $d=6, r=5$ and $d=5, r=2$. They also constructed two families of $k$-optimal LRCs with $d=6, r=2$ and one family with $d=6$ and general locality $r$. Zhang et al. \cite{ZLW,ZTYL} derived two classes of binary $k$-optimal LRCs with $d=6$, $r =11$ and $d\geq6$, $r=2^b(b \geq 3)$ by intersection subspaces, and two classes of ternary $k$-optimal LRCs with $d=6$, $r=3, 5$ using the Kronecker product of two matrices.

This paper focuses on $q$-ary $(n,k,d;r)$-LRCs with disjoint repair sets. An outline of this paper is as follows. In Section 2, we present preliminary concepts related to locally repairable codes and finite projective geometry. In Section 3, we propose a special Pasch configuration, which we denote as the strong Pasch configuration ($s$-Pasch configuration), by exploiting the point-line relationship in $PG(u-1,q)$. Based on the $s$-Pasch configuration, we further provide a geometric characterization for the existence of $q$-ary LRCs with $d\geq5$ and $r=3$. Moreover, we construct two classes of $q$-ary $k$-optimal LRCs with $d=5$ and $r=3$ in $PG(2,q)$. In Section 4, we construct a class of $q$-ary $k$-optimal LRCs with $d=6$ and general locality $r$ by utilizing the structure of partial $r$-spreads. Notably, the coding rates of these LRCs are asymptotically optimal. The concluding remarks are presented in the final section.

\section{Preliminaries}

In this section, we introduce some definitions about locally repairable codes and finite projective geometry. Firstly, we give some notations in this paper.

$\bullet$ For a positive integer $n$, let $[n]=\{1,2,\cdots,n\}$ and let $\binom{[n]}{s}$ be the set of all the $s$-subsets of $[n]$;

$\bullet$ For a prime power $q$, let $\mathbb F_{q}$ be a finite field with $q$ elements, and $\mathbb{F}_{q}^* =\mathbb{F}_{q} \setminus \{0\}$;

$\bullet$ For a positive integer $n$, let $\mathbb{F}_{q}^{n}$ be the $n$-dimensional vector space over $\mathbb F_{q}$;

$\bullet$ For a vector $\boldsymbol{x}=(x_{1},x_{2},\cdots,x_{n})$ in $\mathbb{F}_{q}^{n}$, its support is defined as $\mathrm{supp}(\boldsymbol{x})=\{i\mid x_{i}\neq0,i\in[n]\}$, while its Hamming weight is given by $w(\boldsymbol{x})=|\text{supp}(\boldsymbol{x})|$;

$\bullet$ For a set $\{\boldsymbol{x_i}\mid i\in[l]\}$ consisting of pairwise distinct vectors in $\mathbb{F}_{q}^{n}$, let $\mathrm{Span}\{\boldsymbol{x_i}\mid i\in[l]\}$ be the vector space spanned by $\{\boldsymbol{x_i}\mid i\in[l]\}$;

$\bullet$ A $q$-ary $[n,k,d]$ linear code $\mathcal{C}$ over $\mathbb{F}_{q}$ is a $k$-dimensional subspace of $\mathbb{F}_{q}^{n}$ with minimum distance $d$. The dual code of $\mathcal{C}$ is $\mathcal{C^\bot}=\{u\in\mathbb{F}_{q}^{n}\mid u\cdot v=0, \text{for any}~v\in\mathcal{C}\}$.

\subsection{Locally repairable codes}

This subsection introduces the concept of locally repairable codes and their related properties.

\begin{definition}{\rm\cite{FCX}}\label{1.1}
Let $\mathcal{C}$ be a $q$-ary $[n,k,d]$ linear code. Then $\mathcal{C}$ is called a \emph{locally repairable code} (LRC) \emph{with locality r} if for any~$i\in [n]$, there exists a subset $R_{i}\subset [n]\setminus\{i\}$ with $|R_{i}|\leq r$ such that the $i$-th symbol $c_{i}$ can be represented as a linear combination of ${\{c_{j}\}}_{j\in R_{i}}$, i.e., ${c_{i}}=\sum_{j\in R_{i}} a_{j}c_{j}$, for some ${a_{j}}\in \mathbb{F}_{q}$. We denote such a code as an $(n,k,d;r)$-LRC. The set $R_{i}$ is called a \emph{repair set} for $c_i$.
\end{definition}

Locally repairable codes are typically constructed using generator matrices or parity-check matrices. In the parity-check matrix approach, an equivalent definition can be formulated using the dual code.

\begin{definitionprime}{\rm\cite{WZL}}\label{lrc}
Let $\mathcal{C}$ be a $q$-ary $[n,k,d]$ linear code. Then $\mathcal{C}$ is called a \emph{locally repairable code} (LRC) \emph{with locality r} if for each coordinate $i\in [n]$, there exists a parity-check $\boldsymbol h_{i}\in\mathcal{C^\bot}$ such that $i\in \mathrm{supp}(\boldsymbol h_{i})$ and $w(\boldsymbol h_{i})\leq r+1$. We denote such a code as an $(n,k,d;r)$-LRC. Particularly, a parity-check $\boldsymbol h_{i}$ with $w(\boldsymbol h_{i})\leq r+1$ is called a \emph{local parity-check}. The set $\mathrm{supp}(\boldsymbol h_{i})\setminus\{i\}$ is called a \emph{repair set} for $c_i$.
\end{definitionprime}

\begin{definition}{\rm\cite{WZL}}\label{repair}
Assume that $(r+1)$ divides $n$ and let $l=\frac{n}{r+1}$. An $(n,k,d;r)$-LRC $\mathcal{C}$ is said to have \emph{disjoint repair sets} if there exists local parity-checks $\boldsymbol h_{j_1},\boldsymbol h_{j_2},\cdots,\boldsymbol h_{j_l}\in\mathcal{C^\bot}$ satisfying $w(\boldsymbol h_{j_i})= r+1$ for $i\in [l]$ and $ \mathrm{supp}(\boldsymbol h_{j_i}) \cap \mathrm{supp}(\boldsymbol h_{j_{i'}})=\emptyset$ for~$1\leq i\neq i' \leq l$.
\end{definition}

If $\mathcal{C}$ is a $q$-ary $(n,k,d;r)$-LRC with disjoint repair sets, then the dual space $\mathcal{V}$ of the linear space generated by the local parity-checks $\boldsymbol h_{1},\boldsymbol h_{2},\cdots,\boldsymbol h_{l}$ is called the \emph{$\mathcal{L}$-space} of $\mathcal{C}$. Let $u=n-k-l$. Without loss of generality, we assume that $\mathcal{C}$ has a \emph{parity-check matrix $H$} as the following form
{\small
\begin{equation*}
H=
\left(\begin{array}{cccc|cccc|c|cccc}
1 & 1 & \cdots & 1 & 0 & 0 & \cdots & 0 & \cdots & 0 & 0 & \cdots & 0\\
0 & 0 & \cdots & 0 & 1 & 1 & \cdots & 1 & \cdots & 0 & 0 & \cdots & 0\\
\vdots & \vdots & \ddots & \vdots & \vdots & \vdots & \ddots & \vdots & \cdots & \vdots & \vdots & \ddots & \vdots \\
0 & 0 & \cdots & 0 & 0 & 0 & \cdots & 0 & \cdots & 1 & 1 & \cdots & 1\\
\hline
\boldsymbol v_{0}^{(1)} & \boldsymbol v_{1}^{(1)} & \cdots & \boldsymbol v_{r}^{(1)} & \boldsymbol v_{0}^{(2)} & \boldsymbol v_{1}^{(2)} & \cdots & \boldsymbol v_{r}^{(2)} & \cdots & \boldsymbol v_{0}^{(l)} & \boldsymbol v_{1}^{(l)} & \cdots & \boldsymbol v_{r}^{(l)} \\
\end{array}\right)
=
\left(\begin{array}{c}
H_1\\
\hline
H_2\\
\end{array}\right),
\end{equation*}
}where $H$ is an $(n-k)\times n$ matrix. The upper part (denoted by $H_1$) contains $l$ local parity-checks, which ensure that $\mathcal{C}$ has disjoint repair sets. Clearly, $H_1$ is the parity-check matrix of the $\mathcal{L}$-space $\mathcal{V}$ of $\mathcal{C}$, which implies that $\mathcal{C}$ is a subcode of the $\mathcal{L}$-space $\mathcal{V}$. The lower part (denoted by $H_2$) contains vectors $\boldsymbol v_{j}^{(i)} \in \mathbb{F}_{q}^{u}$ for $i\in [l]$, $0\leq j\leq r$. Under this representation of the parity-check matrix, the $l$ disjoint repair sets of $\mathcal{C}$ are $R_i=\{t\mid (i-1)(r+1)+1\leq t\leq i(r+1)\}$, $1\leq i\leq l$.

The matrix $H$ can also be represented in column vector form as
$$H=(\boldsymbol{h}_{1,0},\boldsymbol{h}_{1,1},\cdots,\boldsymbol {h}_{1,r},\cdots,\boldsymbol{h}_{l,0},\boldsymbol{h}_{l,1},\cdots,\boldsymbol {h}_{l,r}),$$
where $\boldsymbol{h}_{i,j}=\left(
 \begin{array}{c}
 \boldsymbol{e}_{i} \\
 \boldsymbol v_{j}^{(i)} \\
 \end{array}
 \right)\in \mathbb{F}_{q}^{l+u}$, with $\boldsymbol{e}_{i}$ being the $l$-dimensional column vector in $\mathbb{F}_{q}$ whose $i$-th component is $1$ and all other components are $0$, for $i\in [l]$, $0\leq j\leq r$.

The following lemma presents several properties of the parity-check matrix for linear codes, which plays a critical role in determining the minimum distance of such codes.

\begin{lemma}{\rm\cite{LX}}\label{hd}
Let $H$ be a parity-check matrix of a $q$-ary $[n,k]$ linear code $\mathcal{C}$. Then

$(1)$ $\mathcal{C}$ has minimum distance $\geq d$ if and only if any $d-1$ columns of $H$ are linearly independent;

$(2)$ $\mathcal{C}$ has minimum distance $\leq d$ if and only if $H$ has $d$ linearly dependent columns;

$(3)$ $\mathcal{C}$ has minimum distance $d$ if and only if any $d-1$ columns of $H$ are linearly independent and some $d$ columns are linearly dependent.
\end{lemma}

In~\cite{WZL}, Wang et al. derived a Hamming-like upper bound on the dimension $k$ of locally repairable codes using the sphere-packing method.

\begin{lemma}{\rm\cite{WZL}}\label{sphere}
Let $\mathcal{C}$ be a $q$-ary $(n,k,d;r)$-LRC. Then $\mathcal{C}$ satisfies
\begin{equation}\label{2}
k\leq \dim(\mathcal{V})- \log_q\big(B_{\mathcal{V}}(\lfloor \frac{d-1}{2}\rfloor )\big),
\end{equation}
where $B_{\mathcal{V}}(\left\lfloor \frac{d-1}{2}\right\rfloor) = \left| \left\{ \boldsymbol{v} \in \mathcal{V} : w(\boldsymbol{v}) \leq \left\lfloor \frac{d-1}{2}\right\rfloor \right\} \right|$, and $\mathcal{V}$ is the $\mathcal{L}$-space of $\mathcal{C}$.
\end{lemma}

Under the assumption of disjoint repair sets, Wang et al.~\cite{WZL} derived a more explicit expression for bound~\eqref{2} based on properties of the $\mathcal{L}$-space. Furthermore, Fang et al.~\cite{FCX} introduced the terms \emph{perfect LRC} and \emph{$k$-optimal LRC} to denote codes achieving this bound.

\begin{lemma}{\rm\cite{FCX}}{\rm\cite{WZL}}\label{bound1}
Let $\mathcal{C}$ be a $q$-ary $(n,k,d;r)$-LRC with disjoint repair sets. If $l=\frac{n}{r+1}$, then $\mathcal{C}$ satisfies
\begin{equation}\label{3}
q^k \leq \frac{q^{\frac{rn}{r+1}}}{\sum\limits_{0\le i_1+\cdots+i_l \le \left\lfloor \frac{d-1}{2} \right\rfloor} \prod\limits_{j=1}^l \beta(r,i_j)},
\end{equation}
where $\beta(r,i)=\frac{1}{q}\left[(q-1)^i+(-1)^i(q-1)\right]\binom{r+1}{i}$. Moreover, $\mathcal{C}$ is called a \emph{perfect LRC}, if it achieves the bound \eqref{3} with equality.
\end{lemma}

Since $k$ must be an integer, bound~\eqref{3} directly deduces an upper bound on $k$.

\begin{lemma}{\rm\cite{FCX}}{\rm\cite{WZL}}\label{bound2}
Let $\mathcal{C}$ be a $q$-ary $(n,k,d;r)$-LRC with disjoint repair sets. If $l=\frac{n}{r+1}$, then $\mathcal{C}$ satisfies
\begin{equation}\label{4}
k\leq \frac{rn}{r+1} - \left\lceil \log_q\left(\sum\limits_{0\le i_1+\cdots+i_l \le \left\lfloor \frac{d-1}{2} \right\rfloor} \prod\limits_{j=1}^l \beta(r,i_j)\right)\right\rceil,
\end{equation}
where $\beta(r,i)=\frac{1}{q}\left[(q-1)^i+(-1)^i(q-1)\right]\binom{r+1}{i}$. Moreover, $\mathcal{C}$ is called a \emph{$k$-optimal LRC}, if it achieves the bound \eqref{4} with equality.
\end{lemma}

Clearly, the perfect LRC is a special case of the $k$-optimal LRC.

Furthermore, Fang et al.~\cite{FCX} calculated the bounds in Lemmas~\ref{bound1} and~\ref{bound2} for the case where the minimum distance is $5$ or $6$.

\begin{lemma}{\rm\cite{FCX}}\label{pb}
Let $\mathcal{C}$ be a $q$-ary $(n,k,d;r)$-LRC with disjoint repair sets. If $d=5$ or $6$, then
\begin{equation}\label{5}
q^k \leq \frac{q^{\frac{rn}{r+1}}}{\frac{rn}{2}(q-1)+1}
\end{equation}
and
\begin{equation}\label{6}
k\leq \frac{rn}{r+1} - \left\lceil \log_q\left(\frac{rn}{2}(q-1)+1\right)\right\rceil.
\end{equation}
\end{lemma}

\subsection{Finite projective geometry}\label{1.2}

This subsection introduces the concepts concerning finite projective geometry.

\begin{definition}{\rm\cite{CD}}\label{incidence}
If $\mathcal{P}$ is a finite set of points, $\mathcal{L}$ is a finite set of lines, and $I$ is an incidence relation between them, then $P=(\mathcal{P},\mathcal{L},I)$ is called a \emph{finite incidence structure}.
\end{definition}

\begin{definition}{\rm\cite{CD}}\label{projective}
A \emph{finite projective space} $P$ is a finite incidence structure such that

$1$. any two distinct points are on exactly one line;

$2$. let $A,B,C,D$ be four distinct points of which no three are collinear. If the lines $AB$ and $CD$ intersect each other, then the lines $AD$ and $BC$ also intersect each other;

$3$. any line has at least three points.
\end{definition}

\begin{example}{\rm\cite{CD}}\label{pg}
Let $V$ be a vector space of dimension $u$ over the finite field $\mathbb{F}_{q}$, i.e., $V=\mathbb{F}_{q}^{u}$. The geometry $P(V)$ whose points are $1$-dimensional subspaces of $V$ and whose lines are $2$-dimensional subspaces of $V$ is a finite projective space, denoted by $PG(u-1,q)$.
\end{example}

For any nonzero vector $\boldsymbol{x}$ of $\mathbb{F}_{q}^{u}$, let $[\boldsymbol{x}]$ denote the $1$-dimensional subspace generated by $\boldsymbol{x}$, also referred to as a point in $PG(u-1,q)$ with $\boldsymbol{x}$ as its representative vector. For two distinct points $[\boldsymbol{x}]$ and $[\boldsymbol{y}]$ in $PG(u-1,q)$, let $\langle\boldsymbol{x},\boldsymbol{y}\rangle$ denote the line determined by $[\boldsymbol{x}]$ and $[\boldsymbol{y}]$.

Next, we introduce the definitions and properties of $r$-spread and partial $r$-spread, which are critical to one of our constructions of $k$-optimal LRCs.

\begin{definition}{\rm\cite{EJSS}}\label{spread}
A \emph{partial $r$-spread} of $\mathbb{F}_{q}^{u}$ is a collection $S =\{W_1, W_2, \ldots, W_l\}$ of $r$-dimensional subspaces of $\mathbb{F}_{q}^{u}$ such that $W_i \cap W_j =\{\boldsymbol 0\}$ for $1 \leq i < j \leq l$. We call $l$ the \emph{size} of $S$. Moreover, we call $S$ \emph{maximal} if $\left(\mathbb{F}_{q}^{u}\backslash \cup^{l}_{i=1}W_i\right)\cup\{\boldsymbol 0\}$ does not contain an $r$-dimensional subspace, and we call $S$ \emph{maximum} if it has the largest possible size. In particular, if $\cup^{l}_{i=1}W_i=\mathbb{F}_{q}^{u}$, then $S$ is simply called an \emph{$r$-spread}.
\end{definition}

\begin{lemma}{\rm\cite{B}}\label{spreadbound}
If $r\mid u$, then there exists an $r$-spread of $\mathbb{F}_{q}^{u}$ with $l=\frac{q^u-1}{q^r-1}$ subspaces.
\end{lemma}

The distance measure for subspace codes is referred to as the subspace distance, which is typically defined as
$$d_{S}(U,V)=\mathrm{dim} U +\mathrm{dim} V-2\mathrm{dim}(U\cap V),$$
where $U,V$ are two distinct subspaces of $\mathbb{F}_{q}^{u}$. For $k\in [u]$, we denote by $A_q(u,k,d)$ the maximum cardinality of subspace codes over $\mathbb{F}_{q}^{u}$ with minimum subspace distance $d$ and all codewords of dimension $k$. With this notation, the size of a maximum partial $r$-spread in $\mathbb{F}_{q}^{u}$ is given by $A_q(u,r,2r)$.

\begin{lemma}{\rm\cite{EV}}\label{pspreadbound}
Let $u\equiv z\pmod r$. Then, for all $q$, we have
\begin{equation}\label{ps}
A_q(u,r,2r) \geq \frac{q^u-q^r(q^z-1)-1}{q^r-1}.
\end{equation}
\end{lemma}

When $z=0$, there exists an $r$-spread of $\mathbb{F}_{q}^{u}$. Consequently, the inequality \eqref{ps} holds with equality, which is equivalent to Lemma \ref{spreadbound}.

\section{$\boldsymbol{k}$-optimal LRCs with $\boldsymbol{d=5}$ and $\boldsymbol{r=3}$}

In this section, we propose a special Pasch configuration, which we denote as the strong Pasch configuration ($s$-Pasch configuration), by exploiting the point-line relationship in $PG(u-1,q)$. Based on the $s$-Pasch configuration, we then give a geometric characterization for the existence of $q$-ary LRCs with $d\geq5$ and $r=3$. Furthermore, we construct two classes of $q$-ary $k$-optimal LRCs with $d=5$ and $r=3$ in $PG(2,q)$. Throughout this section, we suppose $u\geq 3$ and $4\mid n$. Let $l=\frac{n}{4}$.

We first characterize the existence of $q$-ary LRCs with minimum distance at least $5$ and locality $3$ using the language of vector spaces.

\begin{theorem}\label{v-lrc53}
Suppose $u\geq3$. If there exist nonzero vectors $\boldsymbol v_1^{(i)}, \boldsymbol v_2^{(i)}, \boldsymbol v_3^{(i)}\in\mathbb{F}_{q}^u$ for $i\in[l]$, such that

$(1)$ $\boldsymbol v_1^{(i)}, \boldsymbol v_2^{(i)}, \boldsymbol v_3^{(i)}$ are linearly independent;

$(2)$ let $P_{i}=\{\boldsymbol v_1^{(i)}, \boldsymbol v_2^{(i)}, \boldsymbol v_3^{(i)}, \boldsymbol v_1^{(i)}-\boldsymbol v_2^{(i)}, \boldsymbol v_1^{(i)}-\boldsymbol v_3^{(i)}, \boldsymbol v_2^{(i)}-\boldsymbol v_3^{(i)}\}$, the vectors in $\cup_{i\in[l]} P_{i}$ are pairwise linearly independent, \\
then there exists a $q$-ary $(n=4l,k\geq 3l-u,d\geq 5;r=3)$-LRC with disjoint repair sets.

In particular, if one vector in $P_{i}$ belongs to one of the $2$-dimensional subspaces spanned by vectors in $P_{j}$ for some $1\leq i\neq j\leq l$, then $d=5$.
\end{theorem}

\begin{proof}
Let
{\small
   \begin{equation*}
    H=
    \left(\begin{array}{cccc|cccc|c|cccc}
    1 & 1 & 1 & 1 & 0 & 0 & 0 & 0 & \cdots & 0 & 0 & 0 & 0\\
    0 & 0 & 0 & 0 & 1 & 1 & 1 & 1 & \cdots & 0 & 0 & 0 & 0\\
    \vdots & \vdots & \vdots & \vdots & \vdots & \vdots & \vdots & \vdots & \ddots & \vdots & \vdots & \vdots & \vdots \\
    0 & 0 & 0 & 0 & 0 & 0 & 0 & 0 & \cdots & 1 & 1 & 1 & 1\\
    \hline
    \boldsymbol 0 & \boldsymbol v_{1}^{(1)} & \boldsymbol v_{2}^{(1)} & \boldsymbol v_{3}^{(1)} & \boldsymbol 0 & \boldsymbol v_{1}^{(2)} & \boldsymbol v_{2}^{(2)} & \boldsymbol v_{3}^{(2)} & \cdots & \boldsymbol 0 & \boldsymbol v_{1}^{(l)} & \boldsymbol v_{2}^{(l)} & \boldsymbol v_{3}^{(l)} \\
    \end{array}\right)
    =
    \left(\begin{array}{c}
    H_1\\
    \hline
    H_2\\
    \end{array}\right)
    \end{equation*}
}and let $\mathcal{C}$ be the $q$-ary linear code with parity-check matrix $H$. Then it is easy to see that $\mathcal{C}$ has locality $r=3$, code length $n=4l$ and dimension $k\geq n-l-u=3l-u$. Thus we only need to show that the minimum distance $d\geq5$. By Lemma \ref{hd}, it is sufficient to demonstrate that any four columns of $H$ are linearly independent.

Now express $H$ in column vector form
$$H=(\boldsymbol{h}_{1,0},\boldsymbol{h}_{1,1},\boldsymbol{h}_{1,2},\boldsymbol {h}_{1,3},\cdots,\boldsymbol{h}_{l,0},\boldsymbol{h}_{l,1},\boldsymbol{h}_{l,2},\boldsymbol {h}_{l,3}),$$
where $\boldsymbol{h}_{i,j}=\left(
 \begin{array}{c}
 \boldsymbol{e}_{i} \\
 \boldsymbol v_{j}^{(i)} \\
 \end{array}
 \right)\in \mathbb{F}_{q}^{l+u}$ and $\boldsymbol v_{0}^{(i)}=\boldsymbol 0 \in \mathbb{F}_{q}^{u}$ for any $i\in [l]$, $0\leq j \leq 3$. Note that any column of $H$ cannot be linearly represented by columns from other repair sets. So in order to prove that any four columns in $H$ are linearly independent, we only need to consider any four columns taken from at most two different repair sets. Without loss of generality, we assume these four columns come from the first two repair sets. Then the following two cases should be discussed based on whether the four columns are taken from one or two repair sets.

$(1)$ All four columns are taken from the first repair set. In this case, we suppose
$$a_0 \boldsymbol{h}_{1,0} + a_1 \boldsymbol{h}_{1,1} + a_2 \boldsymbol{h}_{1,2} + a_3 \boldsymbol{h}_{1,3} = \boldsymbol{0},$$
where $a_0, a_1, a_2, a_3 \in \mathbb{F}_q$. Then the following system of equations holds,
\[ \begin{cases}
(a_0+ a_1 + a_2 + a_3 ) \boldsymbol{e}_1= \boldsymbol{0};\\
a_1 \boldsymbol{v}_{1}^{(1)} + a_2 \boldsymbol{v}_{2}^{(1)} + a_3 \boldsymbol{v}_{3}^{(1)} = \boldsymbol{0}.
\end{cases} \]
The linear independence of $\boldsymbol v_1^{(1)},\boldsymbol v_2^{(1)},\boldsymbol v_3^{(1)}$ leads to $a_1 = a_2 = a_3 = 0 $ and hence $a_0 = 0 $. Thus $\boldsymbol{h}_{1,0}, \boldsymbol{h}_{1,1}, \boldsymbol{h}_{1,2}, \boldsymbol{h}_{1,3}$ are linearly independent.

$(2)$ Two columns are taken from the first repair set and two from the second repair set. This can be further divided into the following three cases, according to the number of all-zero columns selected from $H_2$.

Case 1: In this case, assume
$$a_0 \boldsymbol{h}_{1,0} + a_1 \boldsymbol{h}_{1,s} + a_2 \boldsymbol{h}_{2,0} + a_3 \boldsymbol{h}_{2,t} = \boldsymbol{0},$$
where $a_0, a_1, a_2, a_3 \in \mathbb{F}_q$ and $s,t\in[3]$. Then the following system of equations is obtained,
\[ \begin{cases}
(a_0 + a_1) \boldsymbol{e}_1= (a_2 + a_3) \boldsymbol{e}_2= \boldsymbol{0};\\
a_1 \boldsymbol{v}_{s}^{(1)} + a_3 \boldsymbol{v}_{t}^{(2)} = \boldsymbol{0}.
\end{cases} \]
The linear independence of $\boldsymbol v_s^{(1)},\boldsymbol v_t^{(2)}$ leads to $a_1 = a_3 = 0 $ and hence $a_0 =a_2 = 0 $. Thus $\boldsymbol{h}_{1,0}, \boldsymbol{h}_{1,s}, \boldsymbol{h}_{2,0}, \boldsymbol{h}_{2,t}$ are linearly independent.

Case 2: In this case, assume
$$a_0 \boldsymbol{h}_{1,0} + a_1 \boldsymbol{h}_{1,s} + a_2 \boldsymbol{h}_{2,t_1} + a_3 \boldsymbol{h}_{2,t_2} = \boldsymbol{0},$$
where $a_0, a_1, a_2, a_3 \in \mathbb{F}_q$, $s\in[3]$ and $\{t_1,t_2\}\in\binom{[3]}{2}$. Then we have
\[ \begin{cases}
(a_0 + a_1) \boldsymbol{e}_1= (a_2 + a_3) \boldsymbol{e}_2= \boldsymbol{0};\\
a_1 \boldsymbol{v}_{s}^{(1)} + a_2 \boldsymbol{v}_{t_1}^{(2)}+ a_3 \boldsymbol{v}_{t_2}^{(2)} = \boldsymbol{0}.
\end{cases} \]
So $a_3=-a_2$ holds and thus $a_1 \boldsymbol{v}_{s}^{(1)} + a_2 (\boldsymbol{v}_{t_1}^{(2)}- \boldsymbol{v}_{t_2}^{(2)} ) = \boldsymbol{0}$. The linear independence of $\boldsymbol{v}_{s}^{(1)},\boldsymbol{v}_{t_1}^{(2)}- \boldsymbol{v}_{t_2}^{(2)}$ leads to $a_1 = a_2 = 0 $ and hence $a_0 =a_3 = 0 $. Thus $\boldsymbol{h}_{1,0}, \boldsymbol{h}_{1,s}, \boldsymbol{h}_{2,t_1}, \boldsymbol{h}_{2,t_2}$ are linearly independent.

Case 3: In this case, assume
$$a_0 \boldsymbol{h}_{1,s_1} + a_1 \boldsymbol{h}_{1,s_2} + a_2 \boldsymbol{h}_{2,t_1} + a_3 \boldsymbol{h}_{2,t_2} = \boldsymbol{0},$$
where $a_0, a_1, a_2, a_3 \in \mathbb{F}_q$ and $\{s_1,s_2\},\{t_1,t_2\}\in\binom{[3]}{2}$. Then the following system of equations holds,
\[ \begin{cases}
(a_0 + a_1) \boldsymbol{e}_1= (a_2 + a_3) \boldsymbol{e}_2= \boldsymbol{0};\\
a_0 \boldsymbol{v}_{s_1}^{(1)} + a_1 \boldsymbol{v}_{s_2}^{(1)} + a_2 \boldsymbol{v}_{t_1}^{(2)}+ a_3 \boldsymbol{v}_{t_2}^{(2)} = \boldsymbol{0}.
\end{cases} \]
So we have $a_1=-a_0$, $a_3=-a_2$ and thus $a_0 (\boldsymbol{v}_{s_1}^{(1)}-\boldsymbol{v}_{s_2}^{(1)}) + a_2 (\boldsymbol{v}_{t_1}^{(2)}- \boldsymbol{v}_{t_2}^{(2)} ) = \boldsymbol{0}$. The linear independence of $\boldsymbol{v}_{s_1}^{(1)}-\boldsymbol{v}_{s_2}^{(1)}$, $\boldsymbol{v}_{t_1}^{(2)}-\boldsymbol{v}_{t_2}^{(2)}$ leads to $a_0 = a_2 = 0 $ and hence $a_i= 0$, $0\leq i\leq3$. Thus $\boldsymbol{h}_{1,s_1}, \boldsymbol{h}_{1,s_2}, \boldsymbol{h}_{2,t_1}, \boldsymbol{h}_{2,t_2}$ are linearly independent.

From the above analysis, the first conclusion holds.

For the second conclusion, when $\boldsymbol v_s^{(i)}\in \mathrm{Span} \{\boldsymbol v_{1}^{(j)},\boldsymbol v_{2}^{(j)}\}$ for some $i\neq j$ and $s\in [3]$, suppose $\boldsymbol v_s^{(i)}= a \boldsymbol v_{1}^{(j)}+b \boldsymbol v_{2}^{(j)}$ for $a,b\in \mathbb{F}_{q}$. Then we have
$$-\boldsymbol{h}_{i,0}+ \boldsymbol{h}_{i,s}+(a+b)\boldsymbol{h}_{j,0}- a\boldsymbol{h}_{j,1}- b\boldsymbol{h}_{j,2}=\boldsymbol{0}.$$
Therefore, these five columns $\boldsymbol{h}_{i,0}, \boldsymbol{h}_{i,s}, \boldsymbol{h}_{j,0}, \boldsymbol{h}_{j,1}, \boldsymbol{h}_{j,2}$ are linearly dependent. When $\boldsymbol v_{s}^{(i)}-\boldsymbol v_{t}^{(i)}\in \mathrm{Span}\{\boldsymbol v_{1}^{(j)}-\boldsymbol v_{2}^{(j)},\boldsymbol v_{1}^{(j)}-\boldsymbol v_{3}^{(j)}\}$ for some $i\neq j$ and $\{s,t\}\in\binom{[3]}{2}$, suppose $\boldsymbol v_{s}^{(i)}-\boldsymbol v_{t}^{(i)}= a (\boldsymbol v_{1}^{(j)}-\boldsymbol v_{2}^{(j)})+b (\boldsymbol v_{1}^{(j)}-\boldsymbol v_{3}^{(j)})$ for $a,b\in \mathbb{F}_{q}$. Then we have
$$\boldsymbol{h}_{i,s}- \boldsymbol{h}_{i,t}+(a+b)\boldsymbol{h}_{j,1}- a\boldsymbol{h}_{j,2}- b\boldsymbol{h}_{j,3}=\boldsymbol{0}.$$
Therefore, these five columns $\boldsymbol{h}_{i,s}, \boldsymbol{h}_{i,t}, \boldsymbol{h}_{j,1}, \boldsymbol{h}_{j,2}, \boldsymbol{h}_{j,3}$ are linearly dependent. The proof for the other cases can be carried out analogously. Thus we have $d=5$.
\end{proof}

In order to provide a geometric version of Theorem \ref{v-lrc53}, it is necessary to present some geometric configurations.

\begin{definition}{\rm\cite{CD}}\label{pasch}
A \emph{Pasch configuration} is known as a $(6,4)$-configuration in an $STS(v)$, which is a set of $4$ lines whose union contains precisely $6$ points. Here, an $STS(v)$ is a set $V$ of $v$ elements together with a collection $\mathcal{B}$ of $3$-subsets (blocks) of $V$ with the property that every $2$-subset of $V$ occurs in exactly one blocks $B \in \mathcal{B}$.
\end{definition}

Next, we generalize the Pasch configuration in $PG(u-1,q)$ and define a stronger structure.

\begin{definition}\label{s-pasch}
Let $[\boldsymbol v_1]$, $[\boldsymbol v_2]$ and $[\boldsymbol v_3]$ be three non-collinear points in $PG(u-1,q)$. The configuration ($\mathcal{P}$,$\mathcal{L}$) determined by the points in $\mathcal{P}=\{[\boldsymbol v_1], [\boldsymbol v_2], [\boldsymbol v_3], [\boldsymbol v_1-\boldsymbol v_2],[\boldsymbol v_1-\boldsymbol v_3], [\boldsymbol v_2-\boldsymbol v_3]\}$ and the lines in $\mathcal{L}=\{L_1, L_2, L_3, L_4\}$, as shown in Fig \ref{fig:fig1}, is called a \emph{strong Pasch configuration} (\emph{$s$-Pasch configuration}).
\end{definition}

\begin{figure}[htbp]
    \centering  % 关键，让后续内容水平居中
    \begin{tikzpicture}[scale=0.9]{\small
        % 定义直线 l_{i1}
        \draw[thick] (-3,1) -- (3,1) node[right] {$L_{1}$};
        % 定义直线 l_{i2}
        \draw[thick] (-3,1.5) -- (2,-1) node[right] {$L_{2}$};
        % 定义直线 l_{i3}
        \draw[thick] (0,1.75) -- (0,-11/4) node[right] {$L_{3}$};
        % 定义直线
        \draw[thick] (-0.5,-11/4) -- (2.5,7/4) node[right] {$L_{4}$};

        % 标注点和公式，坐标可按需微调
        % p_{i1}
        \draw[fill=black] (-2,1) circle (2pt);
        \node[below left] at (-2,1) {$[\boldsymbol v_1]$};
        % p_{i2}
        \draw[fill=black] (0,1) circle (2pt);
        \node[above left] at (0,1) {$[\boldsymbol v_2]$};
        % p_{i4}
        \draw[fill=black] (2,1) circle (2pt);
        \node[above left] at (2.05,1) {$[\boldsymbol v_1 - \boldsymbol v_2]$};
        % p_{i3}
        \draw[fill=black] (0,0) circle (2pt);
        \node[below left] at (0,0) {$[\boldsymbol v_3]$};
        % p_{i5}
        \draw[fill=black] (1,-0.5) circle (2pt);
        \node[right] at (1.1,-0.4) {$[\boldsymbol v_1 - \boldsymbol v_3]$};
        % p_{i6}
        \draw[fill=black] (0,-2) circle (2pt);
        \node[left] at (0,-1.9) {$[\boldsymbol v_2 - \boldsymbol v_3]$};
        }

    \end{tikzpicture}
    \caption{$s$-Pasch configuration}
    \label{fig:fig1}
\end{figure}
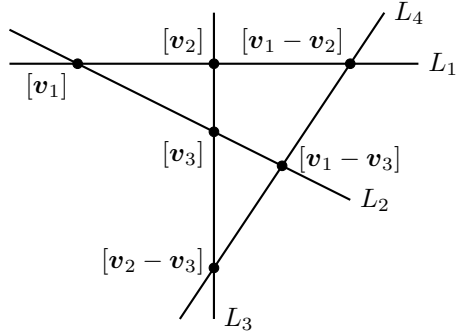

Since the lines in $\mathcal{L}$ are uniquely determined by the points in $\mathcal{P}$, we simply denote the $s$-Pasch configuration ($\mathcal{P}$,$\mathcal{L}$) by $\mathcal{P}$ for brevity.

\begin{definition}\label{disjoint}
Two $s$-Pasch configurations $\mathcal{P}_1$ and $\mathcal{P}_2$ are said to be \emph{disjoint} if their point sets are disjoint.
\end{definition}

It is evident that all the lines determined by any two points in $\mathcal{P}$ include not only the four lines shown in Fig \ref{fig:fig1}, but also the lines $\langle\boldsymbol v_1,\boldsymbol v_2-\boldsymbol v_3\rangle$, $\langle\boldsymbol v_2,\boldsymbol v_1-\boldsymbol v_3\rangle$ and $\langle\boldsymbol v_3,\boldsymbol v_1-\boldsymbol v_2\rangle$. Define $\mathcal{M}=\{\langle\boldsymbol v_1,\boldsymbol v_2-\boldsymbol v_3\rangle, \langle\boldsymbol v_2,\boldsymbol v_1-\boldsymbol v_3\rangle, \langle\boldsymbol v_3,\boldsymbol v_1-\boldsymbol v_2\rangle\}$. Hence, each of the remaining lines in $PG(u-1,q)$ contains at most one point from $\mathcal{P}$.

According to the above analysis, it is possible to restate Theorem \ref{v-lrc53} in geometric terms.

\begin{theoremprime}[1]\label{g-lrc53}
If there exist $l$ pairwise disjoint $s$-Pasch configurations $(\mathcal{P}_{i}$,$\mathcal{L}_{i})$ for $i\in[l]$ in $PG(u-1,q)$, then there exists a $q$-ary $(n=4l,k\geq 3l-u,d\geq 5;r=3)$-LRC with disjoint repair sets.

In particular, if one of the points in $\mathcal{P}_{i}$ lies in one of the lines in $\mathcal{L}_{j}\cup\mathcal{M}_{j}$ for some $1\leq i\neq j\leq l$, then $d=5$.
\end{theoremprime}

\begin{proof}
For any $i\in[l]$, let $\mathcal{P}_i=\{[\boldsymbol v_{1}^{(i)}], [\boldsymbol v_{2}^{(i)}], [\boldsymbol v_{3}^{(i)}], [\boldsymbol v_{1}^{(i)}-\boldsymbol v_{2}^{(i)}], [\boldsymbol v_{1}^{(i)}-\boldsymbol v_{3}^{(i)}], [\boldsymbol v_{2}^{(i)}-\boldsymbol v_{3}^{(i)}]\}$. By the definition of $s$-Pasch configuration, the three points $[\boldsymbol v_{1}^{(i)}]$, $[\boldsymbol v_{2}^{(i)}]$ and $[\boldsymbol v_{3}^{(i)}]$ are not collinear and therefore the representative vectors $\boldsymbol v_1^{(i)},\boldsymbol v_2^{(i)},\boldsymbol v_3^{(i)}$ are linearly independent. From the given hypotheses, these $l$ $s$-Pasch configurations contain precisely $6l$ distinct points from $PG(u-1,q)$. Consequently, it can be deduced that the $6l$ vectors in $\cup_{i\in[l]} P_{i}$ are pairwise linearly independent, where $P_{i}$ is defined in Theorem \ref{v-lrc53}.

Note that one of the points in $\mathcal{P}_{i}$ lying on a line in $\mathcal{L}_{j}\cup\mathcal{M}_{j}$ is equivalent to some vector in $P_{i}$ belonging to one of the $2$-dimensional subspaces spanned by vectors in $P_{j}$ for some $1\leq i\neq j\leq l$.

Then the theorem follows from Theorem \ref{v-lrc53}.
\end{proof}

According to Theorem \ref{g-lrc53}, the existence of $k$-optimal LRCs will be demonstrated via the relationship between points and lines in $PG(u-1,q)$. Next, we construct such an LRC in $PG(2,q)$. We first consider the case where $q$ is an odd prime power, as shown in the following theorem.

\begin{theorem}\label{ko53odd}
Let $q$ be an odd prime power and $q\geq5$, then there exists a $q$-ary $k$-optimal $(4l,3l-3,5;3)$-LRC with $l$ disjoint repair sets, where $\frac{q+1}{6}<l\leq q-2$.
\end{theorem}

\begin{proof}
By Theorem \ref{g-lrc53}, it is sufficient to find $l$ pairwise disjoint $s$-Pasch configurations in $PG(2,q)$, where $\frac{q+1}{6}<l\leq q-2$.

Take a basis $\boldsymbol v_{1}^{(1)}=(1,0,0)$, $\boldsymbol v_{2}^{(1)}=(0,1,0)$, $\boldsymbol v_{3}^{(1)}=(0,0,1)$ of $\mathbb{F}_{q}^3$. Let $\mathcal{P}_1=\{[\boldsymbol v_{1}^{(1)}],[\boldsymbol v_{2}^{(1)}],[\boldsymbol v_{3}^{(1)}],[\boldsymbol v_{1}^{(1)}-\boldsymbol v_{2}^{(1)}],[\boldsymbol v_{1}^{(1)}-\boldsymbol v_{3}^{(1)}],[\boldsymbol v_{2}^{(1)}-\boldsymbol v_{3}^{(1)}]\}$, then $\mathcal{P}_1$ is an $s$-Pasch configuration. Let $M_1=\langle \boldsymbol v_{1}^{(1)},\boldsymbol v_{2}^{(1)}-\boldsymbol v_{3}^{(1)}\rangle$, $M_2=\langle \boldsymbol v_{2}^{(1)},\boldsymbol v_{1}^{(1)}-\boldsymbol v_{3}^{(1)}\rangle$ and $M_3=\langle \boldsymbol v_{3}^{(1)},\boldsymbol v_{1}^{(1)}-\boldsymbol v_{2}^{(1)}\rangle$. The seven lines determined by the points in $\mathcal{P}_1$ and their intersection points are shown in Fig \ref{fig:fig2}.

\begin{figure}[htbp]
    \centering  % 关键，让后续内容水平居中
    \begin{tikzpicture}[scale=1.35]{\small
        % 定义直线 l_{i1}
        \draw[thick] (-5,2) -- (3,2) node[right] {$\langle \boldsymbol v_{1}^{(1)},\boldsymbol v_{2}^{(1)}\rangle$};
        % 定义直线 l_{i2}
        \draw[thick] (-5,2.5) -- (3,-1.5) node[right] {$\langle \boldsymbol v_{1}^{(1)},\boldsymbol v_{3}^{(1)}\rangle$};
        % 定义直线 l_{i3}
        \draw[thick] (0,3) -- (0,-4) node[below] {$\langle \boldsymbol v_{2}^{(1)},\boldsymbol v_{3}^{(1)}\rangle$};
        % 定义直线
        \draw[thick] (18/7,3) -- (-10/7,-4) ;

        \draw[dashed][thick] (-5,23/8) -- (20/7,-4) node[right] {$M_1$};

        \draw[dashed][thick] (-2/7,3) -- (12/7,-4) node[right] {$M_2$};

        \draw[dashed][thick] (3,3) -- (-4,-4) node[right] {$M_3$};

        % 标注点和公式，坐标可按需微调
        % p_{i1}
        \draw[fill=black] (-4,2) circle (2.25pt);
        \node[below left] at (-4,2) {$[\boldsymbol v_1^{(1)}]$};
        % p_{i2}
        \draw[fill=black] (0,2) circle (2.25pt);
        \node[above left] at (-0.05,2) {$[\boldsymbol v_2^{(1)}]$};
        % p_{i4}
        \draw[fill=black] (2,2) circle (2.25pt);
        \node[above left] at (2.05,2) {$[\boldsymbol v_1^{(1)} - \boldsymbol v_2^{(1)}]$};
        % p_{i3}
        \draw[fill=black] (0,0) circle (2.25pt);
        \node[left=4pt] at (0,0) {$[\boldsymbol v_3^{(1)}]$};
        % p_{i5}
        \draw[fill=black] (2/3,-1/3) circle (2.25pt);
        \node[right=5pt] at (2/3,-1/5) {$[\boldsymbol v_1^{(1)} - \boldsymbol v_3^{(1)}]$};
        % p_{i6}
        \draw[fill=black] (0,-3/2) circle (2.25pt);
        \node[left] at (0,-3/2) {$[\boldsymbol v_2^{(1)} - \boldsymbol v_3^{(1)}]$};

        \draw[fill=black] (4/3,-8/3) circle (2.25pt);
        \node[right=3pt] at (4/3,-8/3) {$[\boldsymbol v_1^{(1)}+\boldsymbol v_2^{(1)} - \boldsymbol v_3^{(1)}]$};

        \draw[fill=black] (4/9,4/9) circle (2.25pt);
        \node[right=3pt] at (4/9,4/9) {$[\boldsymbol v_1^{(1)}-\boldsymbol v_2^{(1)}-\boldsymbol v_3^{(1)}]$};

        \draw[fill=black] (-4/5,-4/5) circle (2.25pt);
        \node[left=4pt] at (-4/5,-4/5) {$[\boldsymbol v_1^{(1)}-\boldsymbol v_2^{(1)}+\boldsymbol v_3^{(1)}]$};
        }

    \end{tikzpicture}
    \caption{The seven lines determined by the points in $\mathcal{P}_1$}
    \label{fig:fig2}
\end{figure}
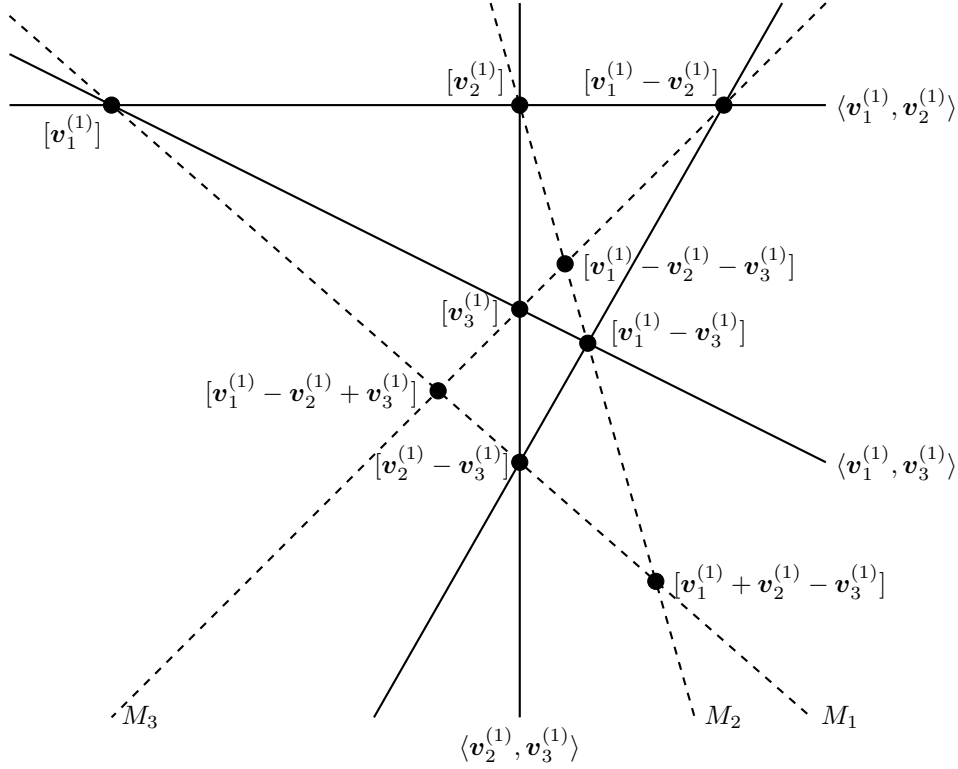

For any $m\in\mathbb{F}_{q}\backslash\{0,1,-1\}$, we take one  point $[\boldsymbol v_i^{(m)}]$ from each $M_i$, $i\in[3]$, with
\[ \begin{cases}
\boldsymbol v_1^{(m)}= \boldsymbol v_1^{(1)}+m(\boldsymbol v_2^{(1)}-\boldsymbol v_3^{(1)}) ;\\
\boldsymbol v_2^{(m)}= m(\boldsymbol v_2^{(1)}+m(\boldsymbol v_1^{(1)}-\boldsymbol v_3^{(1)})) ;\\
\boldsymbol v_3^{(m)}= \boldsymbol v_3^{(1)}+m(\boldsymbol v_1^{(1)}-\boldsymbol v_2^{(1)}).
\end{cases} \]
Since $\boldsymbol v_1^{(m)}-m^{-1}\boldsymbol v_2^{(m)}=(1-m)(\boldsymbol v_1^{(1)}-\boldsymbol v_2^{(1)})$, the line $\langle\boldsymbol v_1^{(m)},\boldsymbol v_2^{(m)}\rangle$ intersects $M_3$ at the point $[\boldsymbol v_1^{(1)}-\boldsymbol v_2^{(1)}]$. It is easy to see that the three points $[\boldsymbol v_1^{(m)}],[\boldsymbol v_2^{(m)}]$ and $[\boldsymbol v_3^{(m)}]$ are non-collinear. Thus, for any $m\in\mathbb{F}_{q}\backslash\{0,1,-1\}$, we can define the $s$-Pasch configuration $\mathcal{P}_m=\{[\boldsymbol v_1^{(m)}],[\boldsymbol v_2^{(m)}],[\boldsymbol v_3^{(m)}],[\boldsymbol v_1^{(m)}-\boldsymbol v_2^{(m)}],$ $[\boldsymbol v_1^{(m)}-\boldsymbol v_3^{(m)}],[\boldsymbol v_2^{(m)}-\boldsymbol v_3^{(m)}]\}$, where
\[ \begin{cases}
\boldsymbol v_1^{(m)}-\boldsymbol v_2^{(m)}= (1-m^2)\boldsymbol v_1^{(1)}+(m^2-m)\boldsymbol v_3^{(1)} ;\\
\boldsymbol v_1^{(m)}-\boldsymbol v_3^{(m)}= (1-m)\boldsymbol v_1^{(1)}+2m\boldsymbol v_2^{(1)}-(m+1)\boldsymbol v_3^{(1)} ;\\
\boldsymbol v_2^{(m)}-\boldsymbol v_3^{(m)}= (m^2-m)\boldsymbol v_1^{(1)}+2m\boldsymbol v_2^{(1)}-(m^2+1)\boldsymbol v_3^{(1)}.
\end{cases} \]

Next we will show that these $q-2$ $s$-Pasch configurations are pairwise disjoint. Firstly, we prove that $\mathcal{P}_1$ is disjoint from $\mathcal{P}_m$ for any $m\in\mathbb{F}_{q}\backslash\{0,1,-1\}$. Note that the choice of $m$ ensures that the coefficients of $\boldsymbol v_1^{(1)}$, $\boldsymbol v_2^{(1)}$ and $\boldsymbol v_3^{(1)}$ in the representative vectors $\boldsymbol v_1^{(m)}$, $\boldsymbol v_2^{(m)}$, $\boldsymbol v_3^{(m)}$ and $\boldsymbol v_1^{(m)}-\boldsymbol v_3^{(m)}$ are all nonzero. It is evident that the coefficients of $\boldsymbol v_1^{(1)}$ and $\boldsymbol v_3^{(1)}$ in the representative vector $\boldsymbol v_1^{(m)}-\boldsymbol v_2^{(m)}$ are not opposite to each other. For the representative vector $\boldsymbol v_2^{(m)}-\boldsymbol v_3^{(m)}$, even though the coefficient of $\boldsymbol v_3^{(1)}$ may be zero, the coefficients of $\boldsymbol v_1^{(1)}$ and $\boldsymbol v_2^{(1)}$ are not opposite to each other. Secondly, we prove that $\mathcal{P}_m$ and $\mathcal{P}_{a}$ are disjoint for all $m,a\in\mathbb{F}_{q}\backslash\{0,1,-1\}$ with $m\neq a$. For further details, please refer to the Appendix \ref{aa}. Thus, we obtain $q-2$ pairwise disjoint $s$-Pasch configurations $\mathcal{P}_i$ for $i\in\mathbb{F}_{q}\backslash\{0,-1\}$.

Finally, it follows from $q\geq5$ that $l\geq 2$. Observe that the point $[\boldsymbol v_1^{(m)} - \boldsymbol v_2^{(m)}]$ in $\mathcal{P}_m$ lies on the line $\langle \boldsymbol v_1^{(1)},\boldsymbol v_3^{(1)}\rangle$ of $\mathcal{L}_1$ for any $m\in\mathbb{F}_{q}\backslash\{0,1,-1\}$. Besides $\mathcal{P}_1$, we can select $l-1$ more from the remaining $q-3$ $s$-Pasch configurations for $2\leq l\leq q-2$. Thus by Theorem \ref{g-lrc53}, there exists a $q$-ary $(n=4l,k\geq3l-3,d=5;r=3)$-LRC with $l$ disjoint repair sets, where $\frac{q+1}{6}<l\leq q-2$. By the upper bound \eqref{6} in Lemma \ref{pb}, the dimension
\begin{equation}\label{7}
k \leq \frac{3n}{4} - \left\lceil \log_q\left(\frac{3n}{2}(q-1) + 1\right)\right\rceil.
\end{equation}
Since $q\geq5$ is an odd prime power and $\frac{q+1}{6}<l\leq q-2$, we have
$$q^2 < \frac{3n}{2}(q-1) + 1 =6l(q-1) +1<q^3.$$
Then from \eqref{7}, it can be seen that $k\leq3l-3$. Thus, $k=3l-3$ and the conclusion of the theorem holds.
\end{proof}

\begin{example}\label{ex2.1}
Let $q=9$ in Theorem \ref{ko53odd}. Suppose $\alpha$ is a root of the primitive polynomial $x^2+x+2$ over $\mathbb{F}_{3}$. Then $\mathbb{F}_9=\{0, 1, 2, \alpha, \alpha+1, \alpha+2, 2\alpha, 2\alpha+1, 2\alpha+2\}$. Now take the following seven $s$-Pasch configurations $\mathcal{P}_m=\{[\boldsymbol v_1^{(m)}],[\boldsymbol v_2^{(m)}],[\boldsymbol v_3^{(m)}],[\boldsymbol v_1^{(m)}-\boldsymbol v_2^{(m)}],[\boldsymbol v_1^{(m)}-\boldsymbol v_3^{(m)}],[\boldsymbol v_2^{(m)}-\boldsymbol v_3^{(m)}]\}$ for $m\in\mathbb{F}_{9}\backslash\{0,-1\}$ in $PG(2,9)$.
{\small
\begin{align*}
\mathcal{P}_1\quad\ ={} &\{[(1,0,0)],[(0,1,0)],[(0,0,1)],[(1,2,0)],[(1,0,2)],[(0,1,2)]\};\\
\mathcal{P}_\alpha\quad\ ={} &\{[(1, \alpha, 2\alpha)],[(2\alpha+1, \alpha, \alpha+2)],[(\alpha , 2\alpha, 1)],\\
{} &[(\alpha, 0, \alpha+1)],[(2\alpha+1, 2\alpha, 2\alpha+2)],[( \alpha+1, 2\alpha, \alpha+1)]\};\\
\mathcal{P}_{\alpha + 1}\ ={} &\{[(1, \alpha + 1, 2\alpha + 2)],[(\alpha+2,  \alpha + 1, 2\alpha+1)],[(\alpha + 1, 2\alpha + 2, 1)],\\
{} &[(2\alpha+2, 0, 1)],[(2\alpha, 2\alpha + 2, 2\alpha+1)],[(1, 2\alpha+2, 2\alpha)]\};\\
\mathcal{P}_{\alpha + 2}\ ={} &\{[(1, \alpha + 2, 2 \alpha + 1)],[(2, \alpha + 2, 1)],[(\alpha + 2, 2\alpha + 1, 1)],\\
{} &[(2, 0, 2\alpha)],[(2\alpha+2, 2\alpha + 1, 2\alpha)],[(2\alpha, 2\alpha + 1, 0)]\};\\
\mathcal{P}_{2\alpha}\quad={} &\{[(1, 2\alpha, \alpha)],[(2\alpha+1, 2\alpha, \alpha+2)],[(2\alpha, \alpha,1)],\\
{} &[(\alpha, 0, 1)],[(\alpha + 1, \alpha, \alpha+2)],[(1, \alpha, \alpha+1)]\};\\
\mathcal{P}_{2\alpha + 1}={} &\{[(1, 2\alpha + 1, \alpha + 2)],[(2, 2\alpha + 1, 1)],[(2\alpha+1, \alpha + 2, 1)],\\
{} &[(2, 0, \alpha+1)],[(\alpha, \alpha + 2, \alpha + 1)],[(\alpha + 1, \alpha+2, 0)]\};\\
\mathcal{P}_{2\alpha + 2}={} &\{[(1, 2\alpha + 2, \alpha + 1)],[(\alpha + 2, 2\alpha+2, 2\alpha + 1)],[(2\alpha + 2, \alpha + 1, 1)],\\
{} &[(2\alpha + 2, 0, 2\alpha)],[(\alpha+2, \alpha + 1, \alpha)],[(2\alpha, \alpha+1, 2\alpha)]\}.
\end{align*}
}It is easy to see that these seven $s$-Pasch configurations are pairwise disjoint and the point $[(\alpha, 0, \alpha+1)]$ in $\mathcal{P}_\alpha$ is on the line $\langle (1,0,0),(0,0,1)\rangle$ in $\mathcal{L}_1$. In addition to $\mathcal{P}_1$ and $\mathcal{P}_\alpha$, select $l-2$ more from the remaining five $s$-Pasch configurations for $2\leq l\leq7$. Thus by Theorem \ref{ko53odd}, we can obtain the existence of a $9$-ary $k$-optimal $(4l,3l-3,5;3)$-LRC with $l$ disjoint repair sets, where $2\leq l\leq7$.
\end{example}

Next, we consider the case where $q$ is an even prime power.

\begin{theorem}\label{ko53even}
Let $q=2^t$, $t\geq 2$, then there exists a $q$-ary $k$-optimal $(4l,3l-3,5;3)$-LRC with $l$ disjoint repair sets, where $\max\{1,\frac{q+1}{6}\}<l\leq q-1$.
\end{theorem}

\begin{proof}
By Theorem \ref{g-lrc53}, it is sufficient to find $l$ pairwise disjoint $s$-Pasch configurations in $PG(2,q)$, where $\max\{1,\frac{q+1}{6}\}<l\leq q-1$.

Take a basis $\boldsymbol v_{1}^{(1)}=(1,0,0)$, $\boldsymbol v_{2}^{(1)}=(0,1,0)$, $\boldsymbol v_{3}^{(1)}=(0,0,1)$ of $\mathbb{F}_{q}^3$. Let $\mathcal{P}_1=\{[\boldsymbol v_{1}^{(1)}],[\boldsymbol v_{2}^{(1)}],[\boldsymbol v_{3}^{(1)}],[\boldsymbol v_{1}^{(1)}+\boldsymbol v_{2}^{(1)}],[\boldsymbol v_{1}^{(1)}+\boldsymbol v_{3}^{(1)}],[\boldsymbol v_{2}^{(1)}+\boldsymbol v_{3}^{(1)}]\}$, then $\mathcal{P}_1$ is an $s$-Pasch configuration. Let $M_1=\langle \boldsymbol v_{1}^{(1)},\boldsymbol v_{2}^{(1)}+\boldsymbol v_{3}^{(1)}\rangle$, $M_2=\langle \boldsymbol v_{2}^{(1)},\boldsymbol v_{1}^{(1)}+\boldsymbol v_{3}^{(1)}\rangle$ and $M_3=\langle \boldsymbol v_{3}^{(1)},\boldsymbol v_{1}^{(1)}+\boldsymbol v_{2}^{(1)}\rangle$. The seven lines determined by the points in $\mathcal{P}_1$ and their intersection points are shown in Fig \ref{fig:fig3}.

\begin{figure}[htbp]
    \centering  % 关键，让后续内容水平居中
    \begin{tikzpicture}[scale=1.35]{\small
        % 定义直线 l_{i1}
        \draw[thick] (-5,2) -- (3,2) node[right] {$\langle \boldsymbol v_{1}^{(1)},\boldsymbol v_{2}^{(1)}\rangle$};
        % 定义直线 l_{i2}
        \draw[thick] (-5,9/4) -- (3,1/4) node[right] {$\langle \boldsymbol v_{1}^{(1)},\boldsymbol v_{3}^{(1)}\rangle$};
        % 定义直线 l_{i3}
        \draw[thick] (0,3) -- (0,-4) node[below, right=1pt] {$\langle \boldsymbol v_{2}^{(1)},\boldsymbol v_{3}^{(1)}\rangle$};
        % 定义直线
        \draw[thick] (5/2,3) -- (-1,-4) ;

        \draw[dashed][thick] (-5,3) -- (2,-4) node[right] {$M_1$};

        \draw[dashed][thick] (1,3) -- (-5,-3) node[right] {$M_2$};

        \draw[dashed][thick] (3,5/2) -- (-5,-3/2) node[right=4pt, below=1pt] {$M_3$};

        % 标注点和公式，坐标可按需微调
        % p_{i1}
        \draw[fill=black] (-4,2) circle (2.25pt);
        \node[below left] at (-4,2) {$[\boldsymbol v_1^{(1)}]$};
        % p_{i2}
        \draw[fill=black] (0,2) circle (2.25pt);
        \node[above left] at (-0.05,2) {$[\boldsymbol v_2^{(1)}]$};
        % p_{i4}
        \draw[fill=black] (2,2) circle (2.25pt);
        \node[above left] at (2.05,2) {$[\boldsymbol v_1^{(1)}+ \boldsymbol v_2^{(1)}]$};
        % p_{i3}
        \draw[fill=black] (0,1) circle (2.25pt);
        \node[left] at (0.05,0.55) {$[\boldsymbol v_3^{(1)}]$};
        % p_{i5}
        \draw[fill=black] (4/3,2/3) circle (2.25pt);
        \node[right=3pt] at (4/3,5/6) {$[\boldsymbol v_1^{(1)}+ \boldsymbol v_3^{(1)}]$};
        % p_{i6}
        \draw[fill=black] (0,-2) circle (2.25pt);
        \node[left=2pt] at (0,-2) {$[\boldsymbol v_2^{(1)}+ \boldsymbol v_3^{(1)}]$};

        \draw[fill=black] (-2,0) circle (2.25pt);
        \node[left=5pt] at (-2,0.1) {$[\boldsymbol v_1^{(1)}+\boldsymbol v_2^{(1)}+\boldsymbol v_3^{(1)}]$};
        }

    \end{tikzpicture}
    \caption{The seven lines determined by the points in $\mathcal{P}_1$}
    \label{fig:fig3}
\end{figure}
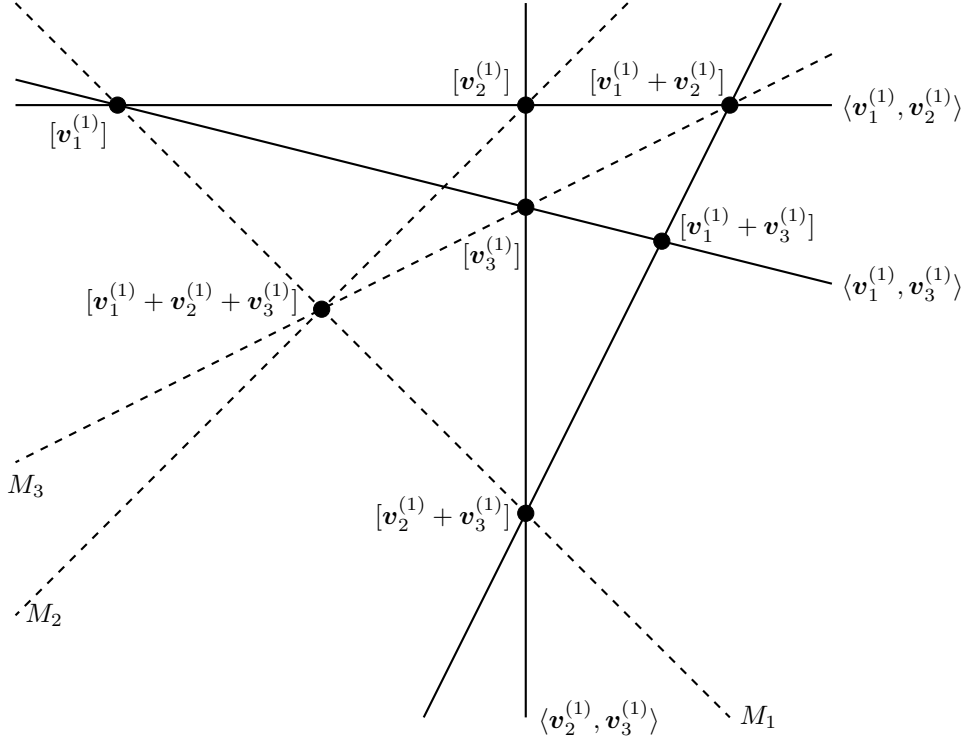

For any $m\in\mathbb{F}_{q}\backslash\{0,1\}$, we take one point $[\boldsymbol v_i^{(m)}]$ from each $M_i$, $i\in[3]$, with
\[ \begin{cases}
\boldsymbol v_1^{(m)}= \boldsymbol v_1^{(1)}+m(\boldsymbol v_2^{(1)}+\boldsymbol v_3^{(1)}) ;\\
\boldsymbol v_2^{(m)}= m(\boldsymbol v_2^{(1)}+m(\boldsymbol v_1^{(1)}+\boldsymbol v_3^{(1)})) ;\\
\boldsymbol v_3^{(m)}= m^2(\boldsymbol v_3^{(1)}+m(\boldsymbol v_1^{(1)}+\boldsymbol v_2^{(1)})).
\end{cases} \]
Since $\boldsymbol v_1^{(m)}+m^{-1}\boldsymbol v_2^{(m)}=(1+m)(\boldsymbol v_1^{(1)}+\boldsymbol v_2^{(1)})$, the line $\langle\boldsymbol v_1^{(m)},\boldsymbol v_2^{(m)}\rangle$ intersects $M_3$ at the point $[\boldsymbol v_1^{(1)}+\boldsymbol v_2^{(1)}]$. It is easy to see that the three points $[\boldsymbol v_1^{(m)}],[\boldsymbol v_2^{(m)}]$ and $[\boldsymbol v_3^{(m)}]$ are non-collinear. Thus, for any $m\in\mathbb{F}_{q}\backslash\{0,1\}$, we can define the $s$-Pasch configuration $\mathcal{P}_m=\{[\boldsymbol v_1^{(m)}],[\boldsymbol v_2^{(m)}],[\boldsymbol v_3^{(m)}],[\boldsymbol v_1^{(m)}+\boldsymbol v_2^{(m)}],$ $[\boldsymbol v_1^{(m)}+\boldsymbol v_3^{(m)}],[\boldsymbol v_2^{(m)}+\boldsymbol v_3^{(m)}]\}$, where
\[ \begin{cases}
\boldsymbol v_1^{(m)}+\boldsymbol v_2^{(m)}= (m^2+1)\boldsymbol v_1^{(1)}+(m^2+m)\boldsymbol v_3^{(1)} ;\\
\boldsymbol v_1^{(m)}+\boldsymbol v_3^{(m)}= (m^3+1)\boldsymbol v_1^{(1)}+m(m^2+1)\boldsymbol v_2^{(1)}+m(m+1)\boldsymbol v_3^{(1)} ;\\
\boldsymbol v_2^{(m)}+\boldsymbol v_3^{(m)}= m^2(m+1)\boldsymbol v_1^{(1)}+m(m^2+1)\boldsymbol v_2^{(1)}.
\end{cases} \]

Next we will show that these $q-1$ $s$-Pasch configurations are pairwise disjoint. Firstly, we prove that $\mathcal{P}_1$ is disjoint from $\mathcal{P}_m$ for any $m\in\mathbb{F}_{q}\backslash\{0,1\}$. Note that the choice of $m$ ensures that the coefficients of $\boldsymbol v_1^{(1)}$, $\boldsymbol v_2^{(1)}$ and $\boldsymbol v_3^{(1)}$ in the representative vectors $\boldsymbol v_1^{(m)}$, $\boldsymbol v_2^{(m)}$, $\boldsymbol v_3^{(m)}$ and $\boldsymbol v_1^{(m)}+\boldsymbol v_3^{(m)}$ are all nonzero. It is evident that the coefficients of $\boldsymbol v_1^{(1)}$ and $\boldsymbol v_3^{(1)}$ in the representative vector $\boldsymbol v_1^{(m)}+\boldsymbol v_2^{(m)}$ are distinct, and the coefficients of $\boldsymbol v_1^{(1)}$ and $\boldsymbol v_2^{(1)}$ in the representative vector $\boldsymbol v_2^{(m)}+\boldsymbol v_3^{(m)}$ are also distinct. Secondly, we prove that $\mathcal{P}_m$ and $\mathcal{P}_{a}$ are disjoint for all $m,a\in\mathbb{F}_{q}\backslash\{0,1\}$ with $m\neq a$. For further details, please refer to the Appendix \ref{ab}. Thus, we obtain $q-1$ pairwise disjoint $s$-Pasch configurations $\mathcal{P}_i$ for $i\in\mathbb{F}_{q}^{*}$.

Finally, it follows from $q\geq4$ that $l\geq2$. It is observed that the point $[\boldsymbol v_1^{(m)} + \boldsymbol v_2^{(m)}]$ in $\mathcal{P}_m$ must be on the line $\langle \boldsymbol v_1^{(1)},\boldsymbol v_3^{(1)}\rangle$ in $\mathcal{L}_1$ for any $m\in\mathbb{F}_{q}\backslash\{0,1\}$. Besides $\mathcal{P}_1$, we can select $l-1$ more from the remaining $q-2$ $s$-Pasch configurations for $2\leq l\leq q-1$. Thus by Theorem \ref{g-lrc53}, there exists a $q$-ary $(n=4l,k\geq3l-3,d=5;r=3)$-LRC with $l$ disjoint repair sets, where $\max\{1,\frac{q+1}{6}\}<l\leq q-1$. By the upper bound \eqref{6} in Lemma \ref{pb}, the dimension satisfies
\begin{equation}\label{8}
k \leq \frac{3n}{4} - \left\lceil \log_q\left(\frac{3n}{2}(q-1) + 1\right)\right\rceil.
\end{equation}
Since $q\geq4$ is an even prime power and $\max\{1,\frac{q+1}{6}\}<l\leq q-1$, we have
$$q^2 < \frac{3n}{2}(q-1) + 1 =6l(q-1) +1<q^3.$$
Then from \eqref{8}, it can be seen that $k\leq3l-3$. Thus, $k=3l-3$ and the theorem's conclusion holds.
\end{proof}

\begin{example}\label{ex2.2}
Let $q=8$ in Theorem \ref{ko53even}. Suppose $\alpha$ is a root of the primitive polynomial $x^3+x+1$ over $\mathbb{F}_{2}$. Then $\mathbb{F}_8=\{0, 1, \alpha, \alpha+1, \alpha^2, \alpha^2+1, \alpha^2+\alpha, \alpha^2+\alpha+1\}$. Now take the following seven $s$-Pasch configurations $\mathcal{P}_m=\{[\boldsymbol v_1^{(m)}],[\boldsymbol v_2^{(m)}],[\boldsymbol v_3^{(m)}],[\boldsymbol v_1^{(m)}+\boldsymbol v_2^{(m)}],$ $[\boldsymbol v_1^{(m)}+\boldsymbol v_3^{(m)}],[\boldsymbol v_2^{(m)}+\boldsymbol v_3^{(m)}]\}$ for $m\in\mathbb{F}_{8}^*$ in $PG(2,8)$.
{\small
\begin{align*}
\mathcal{P}_1\qquad\ ={} &\{[(1,0,0)],[(0,1,0)],[(0,0,1)],[(1,1,0)],[(1,0,1)],[(0,1,1)]\};\\
\mathcal{P}_\alpha\qquad\ ={} &\{[(1, \alpha, \alpha)],[(\alpha^2, \alpha , \alpha^2)],[(\alpha+1, \alpha+1, \alpha^2)],\\
{} &[(\alpha^2+1, 0, \alpha^2+\alpha)],[(\alpha, 1, \alpha^2+\alpha)],[(\alpha^2+\alpha+1, 1, 0)]\};\\
\mathcal{P}_{\alpha + 1}\quad\ ={} &\{[(1, \alpha + 1, \alpha + 1)],[(\alpha^2+1, \alpha + 1, \alpha^2+1)],[(\alpha^2, \alpha^2, \alpha^2+1)],\\
{} &[(\alpha^2, 0, \alpha^2 + \alpha)],[(\alpha^2+1, \alpha^2 + \alpha+1, \alpha^2 + \alpha)],[(1, \alpha^2 + \alpha+1, 0)]\};\\
\mathcal{P}_{\alpha^2}\qquad={} &\{[(1, \alpha^2, \alpha^2)],[(\alpha^2+\alpha, \alpha ^ 2, \alpha^2+\alpha)],[(\alpha^2 + 1, \alpha^2 + 1, \alpha^2+\alpha)],\\
{} &[(\alpha^2 + \alpha+1, 0, \alpha)],[(\alpha^2, 1, \alpha)],[(\alpha + 1,1,0)]\};\\
\mathcal{P}_{\alpha^2+1}\quad={} &\{[(1, \alpha^2 + 1, \alpha^2 + 1)],[(\alpha^2 +\alpha + 1, \alpha^2 + 1, \alpha^2 +\alpha + 1)],[(\alpha^2 +\alpha, \alpha^2 +\alpha,\\
{} & \alpha^2 +\alpha + 1) ],[(\alpha^2 +\alpha, 0, \alpha)], [(\alpha^2 +\alpha + 1, \alpha + 1, \alpha)],[(1, \alpha+1, 0)]\};\\
\mathcal{P}_{\alpha^2+ \alpha}\quad={} &\{[(1, \alpha^2 +\alpha, \alpha^2 +\alpha)],[(\alpha, \alpha^2 +\alpha, \alpha)],[(\alpha^2 +\alpha + 1, \alpha^2 +\alpha + 1, \alpha) ],\\
{} &[(\alpha + 1, 0, \alpha^2) ],[(\alpha^2 +\alpha, 1, \alpha^2)],[(\alpha^2+1, 1, 0)]\};\\
\mathcal{P}_{\alpha^2 +\alpha + 1}={} &\{[(1, \alpha^2 +\alpha + 1, \alpha^2 +\alpha + 1)],[(\alpha+1, \alpha^2 +\alpha + 1, \alpha+1)],[(\alpha, \alpha, \alpha+1)],\\
{} &[(\alpha,0,\alpha^2)],[(\alpha+1, \alpha^2 + 1, \alpha^2) ],[(1, \alpha^2 + 1, 0)]\}.
\end{align*}
}Note that these seven $s$-Pasch configurations are pairwise disjoint and the point $[(\alpha^2+1, 0, \alpha^2+\alpha)]$ in $\mathcal{P}_\alpha$ is on the line $\langle (1,0,0),(0,0,1)\rangle$ in $\mathcal{L}_1$. In addition to $\mathcal{P}_1$ and $\mathcal{P}_\alpha$, select $l-2$ more from the remaining five $s$-Pasch configurations for $2\leq l\leq7$. Thus by Theorem \ref{ko53even}, we can obtain the existence of a $8$-ary $k$-optimal $(4l,3l-3,5;3)$-LRC with $l$ disjoint repair sets, where $2\leq l\leq7$.
\end{example}

\section{$\boldsymbol{k}$-optimal LRCs with $\boldsymbol{d=6}$ and general $\boldsymbol{r}$}

In this section, we formulate our construction within the framework of partial $r$-spreads. For the sake of readability, we first recapitulate the construction of partial $r$-spread presented in \cite{EV} as follows.

Let $u\equiv z\pmod r$, $0\leq z< r$, $m=r+z$ and $u>m$. We may assume without loss of generality that $\alpha$ is a primitive element of $\mathbb{F}_{q^{u-m}}$ and $\beta$ is a primitive element of $\mathbb{F}_{q^{m}}$. Define $t=\frac{q^{u-m}-1}{{q}^r-1}$ and $\gamma=\alpha^t$. It follows that $\gamma$ is a primitive element of $\mathbb{F}_{q^{r}}$. Furthermore, we represent the vectors in $\mathbb{F}_{q}^u$ as
$$\mathbb{F}_{q}^u=\{(x,y)\mid x\in \mathbb{F}_{q^{u-m}},y\in\mathbb{F}_{q^m}\}.$$
Suppose that
\begin{equation}
\begin{cases}\label{S}
R_1=\mathrm{Span}\{(0,\beta^0), (0,\beta^1), \ldots, (0,\beta^{r-1})\};\\
R_{i+2}=\mathrm{Span}\{(\alpha^i,0), (\alpha^i\gamma,0), \ldots, (\alpha^i\gamma^{r-1},0)\};\\
R_{a}=\mathrm{Span}\{(\alpha^i,\beta^j), (\alpha^i\gamma,\beta^{j+1}), \ldots, (\alpha^i\gamma^{r-1},\beta^{j+r-1})\},
\end{cases}
\end{equation}
where $a=(t+1)+(j+1)+i(q^m-1)$, $0\leq i\leq t-1$ and $0\leq j\leq q^m-2$. Then all such subspaces $R_{i}$ form a partial $r$-spread $\mathcal{S}$ with cardinality
$$|\mathcal{S}|=tq^m+1=\frac{q^u-q^r(q^z-1)-1}{q^r-1}\triangleq s.$$
In particular, if $z=0$, then $\mathcal{S}$ is an $r$-spread and $s=\frac{q^u-1}{q^r-1}$.

Based on this construction, we present a method for constructing $k$-optimal LRCs.

\begin{theorem}\label{ko6r}
Let $r\geq 2$, $u\equiv z\pmod r$ and $u\geq 2r$, then there exists a $q$-ary $k$-optimal $(n=(r+1)l,rl-u,6;r)$-LRC with $l$ disjoint repair sets, where $\max\{2,\frac{2(q^{u-1}-1)}{r(r+1)(q-1)}\}< l\leq s$.
\end{theorem}

\begin{proof}
Since $u\equiv z\pmod r$, according to Lemma \ref{pspreadbound}, there exists a partial $r$-spread $\mathcal{S}=\{R_1,R_2,\ldots,R_l\}$, where $\max\{2,\frac{2(q^{u-1}-1)}{r(r+1)(q-1)}\}< l\leq s$ and each $R_i$ is an $r$-subspace of $\mathbb{F}_{q}^u$. By the definition of the partial $r$-spread, we know that $R_i\cap R_j=\{\boldsymbol 0\}$ for $1\leq i\neq j\leq l$. For each $R_i$, $i\in [l]$, we denote its basis which is shown in equation \eqref{S} as $r$ non-zero vectors $\boldsymbol v_{1}^{(i)}, \boldsymbol v_{2}^{(i)}, \ldots, \boldsymbol v_{r}^{(i)}\in\mathbb{F}_{q}^u$. Then we can obtain the following matrix
\begin{equation*}
H=
\left(\begin{array}{cccc|cccc|c|cccc}
1 & 1 & \cdots & 1 & 0 & 0 & \cdots & 0 & \cdots & 0 & 0 & \cdots & 0\\
0 & 0 & \cdots & 0 & 1 & 1 & \cdots & 1 & \cdots & 0 & 0 & \cdots & 0\\
\vdots & \vdots & \ddots & \vdots & \vdots & \vdots & \ddots & \vdots & \cdots & \vdots & \vdots & \ddots & \vdots \\
0 & 0 & \cdots & 0 & 0 & 0 & \cdots & 0 & \cdots & 1 & 1 & \cdots & 1\\
\hline
\boldsymbol 0 & \boldsymbol v_{1}^{(1)} & \cdots & \boldsymbol v_{r}^{(1)} & \boldsymbol 0 & \boldsymbol v_{1}^{(2)} & \cdots & \boldsymbol v_{r}^{(2)} & \cdots & \boldsymbol 0 & \boldsymbol v_{1}^{(l)} & \cdots & \boldsymbol v_{r}^{(l)} \\
\end{array}\right).
\end{equation*}
Let $\mathcal{C}$ be the $q$-ary linear code with parity-check matrix $H$. Then it is easy to see that $\mathcal{C}$ has locality $r$, code length $n=(r+1)l$ and dimension $k\geq n-l-u=rl-u$.

Now express $H$ in column vector form
$$H=(\boldsymbol{h}_{1,0},\boldsymbol{h}_{1,1},\cdots,\boldsymbol {h}_{1,r},\cdots,\boldsymbol{h}_{l,0},\boldsymbol{h}_{l,1},\cdots,\boldsymbol {h}_{l,r}),$$
where $\boldsymbol{h}_{i,j}=\left(
 \begin{array}{c}
 \boldsymbol{e}_{i} \\
 \boldsymbol v_{j}^{(i)} \\
 \end{array}
 \right)\in \mathbb{F}_{q}^{l+u}$, $\boldsymbol v_{0}^{(i)}=\boldsymbol 0 \in \mathbb{F}_{q}^{u}$ for any $i\in [l]$, $0\leq j \leq r$. Since $\boldsymbol v_{1}^{(1)}=(0,\beta^0)$, $\boldsymbol v_{1}^{(2)}=(\alpha^0,0)$, $\boldsymbol v_{1}^{(t+2)}=(\alpha^0,\beta^0)$, we have $\boldsymbol v_{1}^{(1)}+\boldsymbol v_{1}^{(2)}=\boldsymbol v_{1}^{(t+2)}$. From $u\geq 2r$, it can be deduced that $s\geq q^r+1 >3$. Note that $l\geq3$, hence we can always fix the subspaces $R_1$, $R_2$ and $R_{t+2}$ in $\mathcal{S}$. It follows that there always exist six columns $\boldsymbol{h}_{1,0}, \boldsymbol{h}_{1,1}, \boldsymbol{h}_{2,0}, \boldsymbol{h}_{2,1}, \boldsymbol{h}_{t+2,0}, \boldsymbol{h}_{t+2,1}$ in $H$ that are linearly dependent. Thus, by Lemma \ref{hd}, $d\leq 6$.

Note that no column of $H$ can be linearly represented by columns from other repair sets. Thus to prove that any five columns in $H$ are linearly independent, it suffices to consider five columns taken from at most two distinct repair sets. Without loss of generality, we assume these five columns are selected from the first two repair sets. We then divide the discussion into two cases based on how these five columns are selected, and proceed by contradiction.

$(1)$ If there exist five linearly dependent columns in the first repair set, we can deduce that $\boldsymbol v_{1}^{(1)}, \boldsymbol v_{2}^{(1)}, \ldots, \boldsymbol v_{r}^{(1)}$ are linearly dependent. This contradicts the fact that $\boldsymbol v_{1}^{(1)}, \boldsymbol v_{2}^{(1)}, \ldots, \boldsymbol v_{r}^{(1)}$ form a basis of $R_1$.

$(2)$ If there exist five linearly dependent columns drawn from the first and the second repair sets, we can conclude that some nonzero linear combination of $\boldsymbol v_{1}^{(1)}, \boldsymbol v_{2}^{(1)}, \ldots, \boldsymbol v_{r}^{(1)}$ is equal to some nonzero linear combination of $\boldsymbol v_{1}^{(2)}, \boldsymbol v_{2}^{(2)}, \ldots, \boldsymbol v_{r}^{(2)}$. This contradicts the condition $R_1\cap R_2=\{\boldsymbol 0\}$.

Therefore, since any five columns in $H$ are linearly independent, it follows from Lemma \ref{hd} that $d\geq6$, and  we thus obtain $d=6$.

From inequality \eqref{6} in Lemma \ref{pb}, the dimension satisfies
$$k\leq rl - \left\lceil \log_q\left(\frac{r(r+1)l}{2}(q-1)+1\right)\right\rceil.$$
Since $\max\{2,\frac{2(q^{u-1}-1)}{r(r+1)(q-1)}\}< l\leq s$ and $\frac{r(r+1)(q-1)}{2(q^r-1)}\leq 1$, we can calculate that
$$q^{u-1} < \frac{r(r+1)l}{2}(q-1)+1<q^u.$$
That is, $k \leq rl - u$. Thus, the dimension of the code $\mathcal{C}$ is $rl - u$.
\end{proof}

\begin{remark}\label{rem1}
Since $u$ can be arbitrarily large, the code length $n$ of the $k$-optimal LRCs constructed in Theorem \ref{ko6r} is unbounded.
\end{remark}

\begin{remark}\label{rem2}
For fixed $q$, $r$, and $\max\{2,\frac{2(q^{u-1}-1)}{r(r+1)(q-1)}\}< l\leq s$, we derive
\begin{align*}
\lim_{u\rightarrow\infty}\frac{k}{n}&=\lim_{u\rightarrow\infty}\frac{rl-u}{(r+1)l}\\
&=\lim_{u\rightarrow\infty}\left(\frac{r}{r+1}-\frac{u}{(r+1)l}\right)\\
&=\frac{r}{r+1}.
\end{align*}
Therefore, the code rates of these $k$-optimal LRCs constructed in Theorem \ref{ko6r} are asymptotically optimal.
\end{remark}

\begin{remark}\label{rem3}
For $z=0$, the $k$-optimal LRCs with the same parameters are also constructed in \cite{FCX} via $q$-Steiner systems. Specifically, when $q=r=2$ and $l=\frac{2^u-1}{3}$, this corresponds to a binary $k$-optimal $(3l,2l-u,6;2)$-LRC with $l$ disjoint repair sets. An LRC with the same parameters was also constructed in \cite{FCX} using difference set and in \cite{GC} using cyclic codes.
\end{remark}

\begin{landscape}  % 横向页面开始
\begin{table}[htbp]
%\small %字体大小的控制
\centering
\caption{$k$-Optimal LRCs with disjoint repair sets}
\label{tab:k}
\renewcommand\arraystretch{1.8}
\setlength{\tabcolsep}{7pt} % 调整列间距
\begin{tabular}{ccccccccc}
\toprule  % 顶部专业线
\textbf{No.} & \textbf{\textit{q}} & \textbf{\textit{n}} & \textbf{\textit{k}} & \textbf{\textit{d}} & \textbf{\textit{r}} & \textbf{\textit{u}} & \textbf{\textit{l}} & \textbf{Ref.} \\
\midrule  % 中部专业线

1 & $q>2$, $3\nmid q$ & $3l$ & $2l-u$ & 5 & 2 & $u\geq 2$ & $\frac{q^{u} - 1}{3(q-1)}$ & \cite{FCX} \\

\addlinespace[3pt] % 增加行间距
2 & $q\geq 5$, $2\nmid q$ & $4l$ & $3l-u$ & 5 & 3 & 3 & $\frac{q+1}{6}<l\leq q-2$ & Theorem~\ref{ko53odd} \\

\addlinespace[3pt]
3 & $q=2^t$, $t\geq 2$ & $4l$ & $3l-u$ & 5 & 3 & 3 & $\max\{1,\frac{q+1}{6}\}<l\leq q-1$ & Theorem~\ref{ko53even} \\

\midrule  % 组间分隔线
\addlinespace[5pt] % 组间额外间距

4 & 2 & $(r+1)l$ & $rl-u$ & 6 & $2^b$ & $u \geq 4b$ & $\frac{2^{u-1}-1}{2^{b-1}(2^b+1)} < l \leq A_2(u,2b,4b)$ & \cite{MG} \\[8pt]

\addlinespace[3pt] % 增加行间距
5 & 2 & $4l$ & $3l-u$ & 6 & 3 & $u \geq 6$ & $\frac{2^{u-1} - 1}{6} < l \leq A_2(u,3,6)$ & \cite{MG} \\[8pt]

\addlinespace[3pt]
6 & 2 & $4l$ & $3l-u$ & 6 & 3 & $u \geq 5$ & $\left\lfloor\frac{2^{u-2}-1}{3}\right\rfloor + 1 \leq l \leq A_2(u-1,2,4)$ & \cite{MG} \\[8pt]

\addlinespace[3pt]
7 & 2 & $12l$ & $11l-u$ & 6 & 11 & $6\mid u-1$ & $\frac{2^{u-1} - 1}{2^{6}-1}$ & \cite{ZLW} \\

\addlinespace[3pt]
8 & 3 & $4l$ & $3l-u$ & 6 & 3 & $2\mid u-1$ & $\frac{3^{u-1}-1}{3^2-1}$ & \cite{ZLW} \\

\addlinespace[3pt]
9 & 3 & $6l$ & $5l-u$ & 6 & 5 & $3\mid u-1$ & $\frac{3^{u-1}-1}{3^3-1}$ & \cite{ZLW} \\

\addlinespace[3pt]
10 & $q$ & $(r+1)l$ & $rl-u$ & 6 & $r$ & $\makecell{2\leq r< u,~r\mid u,\\ \frac{q^{u}-1}{q^{u-1}}\geq\frac{2(q^r-1)}{(r+1)r(q-1)}}$ & $\frac{q^{u} - 1}{q^{r}-1}$ & \cite{FCX} \\

\addlinespace[3pt]
11 & $q$ & $(r+1)l$ & $rl-u$ & 6 & $r$ & $\makecell{u\geq 2r,\\ u\equiv z\pmod r}$ & $\max\{2,\frac{2(q^{u-1}-1)}{r(r+1)(q-1)}\}< l\leq \frac{q^u-q^r(q^z-1)-1}{q^r-1}$ & Theorem~\ref{ko6r} \\

\addlinespace[2pt]
\bottomrule  % 底部专业线
\end{tabular}
\end{table}
\end{landscape}  % 横向页面结束

\begin{example}\label{ex3.1}
Let $q=r=3$ in Theorem \ref{ko6r}.

(1) If $u\equiv 0\pmod 3$ and $u\geq6$, then there exists a ternary $k$-optimal $(4l,3l-u,6;3)$-LRC with $l$ disjoint repair sets, where $\frac{3^{u-1}-1}{12}< l\leq \frac{3^{u}-1}{26}$. This includes Example 5 in \textup{\cite{FCX}}, which corresponds to $u=6$ and $l=28$.

(2) If $u\equiv 1\pmod 3$ and $u\geq7$, then there exists a ternary $k$-optimal $(4l,3l-u,6;3)$-LRC with $l$ disjoint repair sets, where $\frac{3^{u-1}-1}{12}< l\leq \frac{3^{u}-55}{26}$.

(3) If $u\equiv 2\pmod 3$ and $u\geq8$, then there exists a ternary $k$-optimal $(4l,3l-u,6;3)$-LRC with $l$ disjoint repair sets, where $\frac{3^{u-1}-1}{12}< l\leq \frac{3^{u}-217}{26}$.
\end{example}

\section{Conclusion}

We summarize the existing known results and our conclusions for $k$-optimal LRCs with disjoint repair sets in Table \ref{tab:k}.

In this paper, we investigate locally repairable codes (LRCs) with disjoint repair sets via the parity-check matrix method. First, we generalize the Pasch configuration to propose the novel concept of an $s$-Pasch configuration. We then derive a geometric characterization for the existence of LRCs with $d=5$ and $r=3$. Subsequently, we construct several $k$-optimal LRCs by employing the point-line relationship in $PG(2,q)$. It is of great interest to find more pairwise disjoint $s$-Pasch configurations in $PG(u-1,q)$ to construct $k$-optimal LRCs with a greater number of disjoint repair sets.

Moreover, a family of $q$-ary $k$-optimal LRCs with $d=6$ and general $r$ is constructed by using partial $r$-spreads. It is noteworthy that both the range of possible code lengths and the number of repair sets for such LRCs are expanded compared to the results presented by Fang et al. in \cite{FCX}. An open question remains: whether additional geometric structures can be developed to construct $k$-optimal LRCs with larger $d$. This direction warrants further exploration.

\vspace{0.2cm}

\noindent {\bf Acknowledgements}

This study is supported by Science Research Project of Hebei Education Department under Grant JCZX2026032, Natural Science Foundation of Hebei Province under Grant A2025205023 and A2023205045.

\appendix
\renewcommand{\thesection}{\Alph{section}}
\titleformat{\section}{\Large\bfseries}{Appendix \thesection}{1em}{}

\section{The partial proof of Theorem \ref{ko53odd}}\label{aa}

In the following, we verify that for all $m,a\in\mathbb{F}_{q}\backslash\{0,1,-1\}$ with $m\neq a$, $\mathcal{P}_m$ and $\mathcal{P}_{a}$ are disjoint.

Note that the three lines $M_1,M_2$ and $M_3$ intersect pairwise at three distinct points, and all remaining points on these lines are mutually distinct. Then it follows that $\{[\boldsymbol v_1^{(m)}],[\boldsymbol v_2^{(m)}],[\boldsymbol v_3^{(m)}]\} \cap \{[\boldsymbol v_1^{(a)}],[\boldsymbol v_2^{(a)}],[\boldsymbol v_3^{(a)}]\}=\emptyset$.

It is evident from $a\neq0$ that the coefficient of $\boldsymbol v_2^{(1)}$ in the representative vectors of all five points in $\mathcal{P}_{a}$ other than $[\boldsymbol v_1^{(a)}-\boldsymbol v_2^{(a)}]$ is nonzero. Therefore, $[\boldsymbol v_1^{(m)}-\boldsymbol v_2^{(m)}]\not\in \mathcal{P}_{a}\backslash\{[\boldsymbol v_1^{(a)}-\boldsymbol v_2^{(a)}]\}$. Since $m,a\in\mathbb{F}_{q}\backslash\{0,1,-1\}$ and $m\neq a$,
it holds
$$\frac{1-m^2}{1-a^2}\neq\frac{m^2-m}{a^2-a}.$$
This implies that the coefficients of $\boldsymbol v_1^{(1)}$ and $\boldsymbol v_3^{(1)}$ in the representative vectors $\boldsymbol v_1^{(m)}-\boldsymbol v_2^{(m)}$ and $\boldsymbol v_1^{(a)}-\boldsymbol v_2^{(a)}$ are not proportional, then we have $[\boldsymbol v_1^{(m)}-\boldsymbol v_2^{(m)}]\neq[\boldsymbol v_1^{(a)}-\boldsymbol v_2^{(a)}]$. Thus, $[\boldsymbol v_1^{(m)}-\boldsymbol v_2^{(m)}]\not\in \mathcal{P}_{a}$. Similarly, we obtain that $[\boldsymbol v_1^{(a)}-\boldsymbol v_2^{(a)}]\not\in\mathcal{P}_m$.

Since the coefficients of any two of $\boldsymbol v_1^{(1)}$, $\boldsymbol v_2^{(1)}$ and $\boldsymbol v_3^{(1)}$ in the representative vector $\boldsymbol v_1^{(m)}-\boldsymbol v_3^{(m)}$ are not opposite numbers, it follows that $[\boldsymbol v_1^{(m)}-\boldsymbol v_3^{(m)}]\not\in\{[\boldsymbol v_1^{(a)}],[\boldsymbol v_2^{(a)}],[\boldsymbol v_3^{(a)}]\}$. From the choice of $m$ and $a$, it follows that $[\boldsymbol v_1^{(m)}-\boldsymbol v_3^{(m)}]\not\in\{[\boldsymbol v_1^{(a)}-\boldsymbol v_2^{(a)}],[\boldsymbol v_1^{(a)}-\boldsymbol v_3^{(a)}],[\boldsymbol v_2^{(a)}-\boldsymbol v_3^{(a)}]\}$. Thus, $[\boldsymbol v_1^{(m)}-\boldsymbol v_3^{(m)}]\not\in \mathcal{P}_{a}$. By the same token, we obtain $[\boldsymbol v_1^{(a)}-\boldsymbol v_3^{(a)}]\not\in\mathcal{P}_m$.

In a similar manner, it can be verified that $[\boldsymbol v_2^{(m)}-\boldsymbol v_3^{(m)}]\not\in\mathcal{P}_a$ and  $[\boldsymbol v_2^{(a)}-\boldsymbol v_3^{(a)}]\not\in\mathcal{P}_m$.

Based on the above analysis, it follows that for all $m,a\in\mathbb{F}_{q}\backslash\{0,1,-1\}$ with $m\neq a$, $\mathcal{P}_m$ and $\mathcal{P}_{a}$ are disjoint.

\section{The partial proof of Theorem \ref{ko53even}}\label{ab}

In the following, we verify that for all $m,a\in\mathbb{F}_{q}\backslash\{0,1\}$ with $m\neq a$, $\mathcal{P}_m$ and $\mathcal{P}_{a}$ are disjoint.

Note that the three lines $M_1,M_2$ and $M_3$ intersect at a single point, and all remaining points on these lines are mutually distinct. This implies $\{[\boldsymbol v_1^{(m)}],[\boldsymbol v_2^{(m)}],[\boldsymbol v_3^{(m)}]\}\cap \{[\boldsymbol v_1^{(a)}],[\boldsymbol v_2^{(a)}],[\boldsymbol v_3^{(a)}]\}=\emptyset$.

It is evident from $a\neq0,1$ that the coefficient of $\boldsymbol v_2^{(1)}$ in the representative vectors of all five points in $\mathcal{P}_{a}$ other than $[\boldsymbol v_1^{(a)}+\boldsymbol v_2^{(a)}]$ is nonzero. Therefore, $[\boldsymbol v_1^{(m)}+\boldsymbol v_2^{(m)}]\not\in\mathcal{P}_{a}\backslash\{[\boldsymbol v_1^{(a)}+\boldsymbol v_2^{(a)}]\}$. Since $m,a\in\mathbb{F}_{q}\backslash\{0,1\}$ and $m\neq a$, it holds
$$\frac{m^2+1}{a^2+1}\neq\frac{m^2+m}{a^2+a}.$$
This implies that the coefficients of $\boldsymbol v_1^{(1)}$ and $\boldsymbol v_3^{(1)}$ in the representative vectors $\boldsymbol v_1^{(m)}+\boldsymbol v_2^{(m)}$ and $\boldsymbol v_1^{(a)}+\boldsymbol v_2^{(a)}$ are not proportional, then we have $[\boldsymbol v_1^{(m)}+\boldsymbol v_2^{(m)}]\neq[\boldsymbol v_1^{(a)}+\boldsymbol v_2^{(a)}]$. Thus, $[\boldsymbol v_1^{(m)}+\boldsymbol v_2^{(m)}]\not\in \mathcal{P}_{a}$. Similarly, we obtain that $[\boldsymbol v_1^{(a)}+\boldsymbol v_2^{(a)}]\not\in\mathcal{P}_m$.

Moreover, it can be similarly verified that $[\boldsymbol v_2^{(m)}+\boldsymbol v_3^{(m)}]\not\in\mathcal{P}_{a}$ and $[\boldsymbol v_2^{(a)}+\boldsymbol v_3^{(a)}]\not\in\mathcal{P}_{m}$.

Since the coefficients of any two of $\boldsymbol v_1^{(1)}$, $\boldsymbol v_2^{(1)}$ and $\boldsymbol v_3^{(1)}$ in the representative vector $\boldsymbol v_1^{(m)}+\boldsymbol v_3^{(m)}$ are distinct, it follows that $[\boldsymbol v_1^{(m)}+\boldsymbol v_3^{(m)}]\not\in\{[\boldsymbol v_1^{(a)}],[\boldsymbol v_2^{(a)}],[\boldsymbol v_3^{(a)}]\}$. From the choice of $m$ and $a$, it follows that $[\boldsymbol v_1^{(m)}+\boldsymbol v_3^{(m)}]\not\in\{[\boldsymbol v_1^{(a)}+\boldsymbol v_2^{(a)}],[\boldsymbol v_1^{(a)}+\boldsymbol v_3^{(a)}],[\boldsymbol v_2^{(a)}+\boldsymbol v_3^{(a)}]\}$. Thus, $[\boldsymbol v_1^{(m)}+\boldsymbol v_3^{(m)}]\not\in \mathcal{P}_{a}$. By the same token, we obtain $[\boldsymbol v_1^{(a)}+\boldsymbol v_3^{(a)}]\not\in\mathcal{P}_m$.

Based on the above analysis, it follows that for all $m,a\in\mathbb{F}_{q}\backslash\{0,1\}$ with $m\neq a$, $\mathcal{P}_m$ and $\mathcal{P}_{a}$ are disjoint.


\begin{thebibliography}{99}

\bibitem{B}
Bu T.: Partitions of a vector space. Discret. Math. \textbf{31}(1), 79-83 (1980).

\bibitem{CM}
Cadambe V.R., Mazumdar A.: Bounds on the size of locally recoverable codes. IEEE Trans. Inf. Theory \textbf{61}(11), 5787-5794 (2015).

\bibitem{CCFT}
Cai H., Cheng M., Fan C., Tang X.: Optimal locally repairable systematic codes based on packings. IEEE Trans. Commun. \textbf{67}(1), 39-49 (2019).

\bibitem{CMST}
Cai H., Miao Y., Schwartz M., Tang X.: On optimal locally repairable codes with multiple disjoint repair sets. IEEE Trans. Inf. Theory \textbf{66}(4), 2402-2416 (2019).

\bibitem{CFXH}
Chen B., Fang W., Xia S., Hao J., Fu F.: Improved bounds and Singleton-optimal constructions of locally repairable codes with minimum distance 5 and 6. IEEE Trans. Inf. Theory \textbf{67}(1), 217-231 (2021).

\bibitem{CD}
Colbourn C.J., Dinitz J.H.: Handbook of combinatorial designs, 2nd Edn (Discrete Mathematics and Its Applications). Chapman \& Hall/CRC (2007).

\bibitem{EJSS}
El-Zanati S., Jordon H., Seelinger G., Sissokho P., Spence L.: The maximum size of a partial 3-spread in a finite vector space over GF(2). Des. Codes Cryptogr. \textbf{54}(2), 101-107 (2010).

\bibitem{EV}
Etzion T., Vardy A.: Error-correcting codes in projective space. IEEE Trans. Inf. Theory \textbf{57}(2), 1165-1173 (2011).

%\bibitem{ER}
%Etzion T, Raviv N. Equidistant codes in the Grassmannian. Discret. Appl. Math. 186,87-97 (2015).

\bibitem{FCX}
Fang W., Chen B., Xia S., Fu F., Chen X.: Perfect LRCs and $k$-optimal LRCs. Des. Codes Cryptogr. \textbf{91}(4), 1209-1232 (2023).

\bibitem{FF}
Fang W., Fu F.: Optimal cyclic $(r, \delta)$ locally repairable codes with unbounded length. Finite Fields Appl. \textbf{63}, 101650, (2020).

\bibitem{FFCX}
Fang W., Fu F., Chen B., Xia S.: Singleton-optimal LRCs and perfect LRCs via cyclic and constacyclic codes. Finite Fields Appl. \textbf{91}, 102273 (2023).

\bibitem{FTFC}
Fang W., Tao R., Fu F., Chen B., Xia S.: Bounds and constructions of singleton-optimal locally repairable codes with small localities. IEEE Trans. Inf. Theory \textbf{70}(10), 6842-6856 (2024).

\bibitem{GC}
Goparaju S., Calderbank R.: Binary cyclic codes that are locally repairable. In: 2014 IEEE International Symposium on Information Theory (ISIT), pp. 676-680 (2014).

\bibitem{GHSY}
Gopalan P., Huang C., Simitci H., Yekhanin S.: On the locality of codeword symbols. IEEE Trans. Inf. Theory \textbf{58}(11), 6925-6934 (2012).

\bibitem{GXY}
Guruswami V., Xing C., Yuan C.: How long can optimal locally repairable codes be? IEEE Trans. Inf. Theory \textbf{65}(6), 3662-3670 (2019).

\bibitem{HCL}
Huang C., Chen M., Li J.: Pyramid codes: flexible schemes to trade space for access efficiency in reliable data storage systems. In: sixth IEEE International Symposium on Network Computing and Applications (NCA), pp. 79-86 (2007).

\bibitem{HHCY}
Han X., Han G., Cai H., Yin L.: Locally repairable codes with multiple repair sets based on packings of block size 4. Cryptogr. Commun. \textbf{16}, 459-479 (2023).

\bibitem{HXSC}
Hao J., Xia S., Shum K.W., Chen B., Fu F., Yang Y.: Bounds and constructions of locally repairable codes: Parity-check matrix approach. IEEE Trans. Inf. Theory \textbf{66}(12), 7465-7474 (2020).

\bibitem{JC}
Jiang J., Cheng M.: Regular $(k, R, 1)$-packings with max$(R)= 3$ and their locally repairable codes. Cryptogr. Commun. \textbf{12}, 1071-1089 (2020).

\bibitem{JMX}
Jin L., Ma L., Xing C.: Construction of optimal locally repairable codes via automorphism groups of rational function fields. IEEE Trans. Inf. Theory \textbf{66}(1), 210-221 (2020).

\bibitem{KN1}
Kim C., No J.S.: New constructions of binary and ternary locally repairable codes using cyclic codes. IEEE Commun. Lett. \textbf{22}(2), 228-231 (2017).

\bibitem{KN2}
Kim C., No J.S.: New constructions of binary LRCs with disjoint repair groups and locality 3 using existing LRCs. IEEE Commun. Lett. \textbf{23}(3), 406-409 (2019).

\bibitem{KWG}
Kong X., Wang X., Ge G.: New constructions of optimal locally repairable codes with super-linear length. IEEE Trans. Inf. Theory \textbf{67}(10), 6491-6506 (2021).

\bibitem{LMX}
Li X., Ma L., Xing C.: Optimal locally repairable codes via elliptic curves. IEEE Trans. Inf. Theory \textbf{65}(1), 108-117 (2019).

\bibitem{LX}
Ling S., Xing C.: Coding Theory: A First Course. Cambridge: Cambridge University Press (2004).

\bibitem{LXY}
Luo Y., Xing C., Yuan C.: Optimal locally repairable codes of distance 3 and 4 via cyclic codes. IEEE Trans. Inf. Theory \textbf{65}(2), 1048-1053 (2019).

\bibitem{LYRF}
Li R., Yang S., Rao Y., Fu Q.: On binary locally repairable codes with distance four. Finite Fields Appl. \textbf{72}, 101793 (2021).

\bibitem{MG}
Ma J., Ge G.: Optimal binary linear locally repairable codes with disjoint repair groups. SIAM J. Discret. Math. \textbf{33}(4), 2509-2529 (2019).

\bibitem{PKLK}
Prakash N., Kamath G.M., Lalitha V., Kumar P.V.: Optimal linear codes with a local-error-correction property. In: 2012 IEEE International Symposium on Information Theory (ISIT), pp. 2776-2780 (2012).

\bibitem{QFF}
Qiu J., Fang W., Fu F.: New lower bounds for the minimum distance of cyclic codes and applications to locally repairable codes. IEEE Trans. Inf. Theory \textbf{70}(7), 4968-4982 (2024).

\bibitem{QZF}
Qiu J., Zheng D., Fu F.: New constructions of optimal cyclic $(r,\delta)$ locally repairable codes from their zeros. IEEE Trans. Inf. Theory \textbf{67}(3), 1596-1608 (2021).

\bibitem{RPDV}
Rawat A.S., Papailiopoulos D.S., Dimakis A.G, Vishwanath S.: Locality and availability in distributed storage. IEEE Trans. Inf. Theory \textbf{62}(8), 4481-4493 (2016).

\bibitem{SZ}
Silberstein N., Zeh A.: Optimal binary locally repairable codes via anticodes. In: 2015 IEEE International Symposium on Information Theory (ISIT), pp. 1247-1251 (2015).

\bibitem{TB}
Tamo I., Barg A.: A family of optimal locally recoverable codes. IEEE Trans. Inf. Theory \textbf{60}(8), 4661-4676 (2014).

\bibitem{TFWF}
Tao R., Fang W., Wang Y., Fu F., Hu S.: Some new results on improved bounds and constructions of singleton-optimal $(r,\delta)$ locally repairable codes. IEEE Trans. Commun.  \textbf{73}(5), 2876-2890 (2024).

\bibitem{TZ}
Tian F., Zhu S.: Optimal cyclic locally repairable codes with unbounded length from their zeros. Discret. Math. \textbf{345}(9), 112978 (2022).

\bibitem{TZSP}
Tan P., Zhou Z., Sidorenko V., Parampalli U.: Two classes of optimal LRCs with information $(r, t)$-locality. Des. Codes Cryptogr. \textbf{88}(9), 1741-1757 (2020).

\bibitem{WZ}
Wang A., Zhang Z.: Repair locality with multiple erasure tolerance. IEEE Trans. Inf. Theory \textbf{60}(11), 6979-6987 (2014).

\bibitem{WZL}
Wang A., Zhang Z., Lin D.: Bounds for binary linear locally repairable codes via a sphere-packing approach. IEEE Trans. Inf. Theory \textbf{65}(7), 4167-4179 (2019).

\bibitem{ZK}
Zhang Y., Kan H.: Locally Repairable Codes from Combinatorial Designs. Sci. China Inf. Sci. \textbf{63}(2), 122304 (2020).

\bibitem{ZLS}
Zhao J., Li Y., Shan X.: Optimal locally repairable codes with multiple repair sets based on $2$-regular packings. Discret. Math. \textbf{348}(9), 114499 (2025).

\bibitem{ZLW}
Zhang W., Luo Y., Wang L.: Optimal binary and ternary locally repairable codes with minimum distance 6. Des. Codes Cryptogr. \textbf{92}(5), 1251-1265 (2024).

\bibitem{ZTYL}
Zhang W., Tang D., Ying C., Luo Y.: Constructions of optimal binary locally repairable codes via intersection subspaces. Sci. China Inf. Sci. \textbf{67}(6), 162306 (2024).

\end{thebibliography}
\end{document}